\documentclass[twocolumn, final, 10pt]{IEEEtran}
\usepackage{algorithm}
\usepackage{algpseudocode}
\usepackage[dvips]{graphicx}
\usepackage{times}
\usepackage{cite}
\usepackage{amsmath}
\usepackage{array}
\usepackage{amssymb}
\usepackage{color}

\usepackage{bbm}

\usepackage{float}
\usepackage{stfloats}
\usepackage{rotating,threeparttable,booktabs}
\usepackage{bm}
\usepackage{dcolumn,booktabs}
\usepackage{multirow}
\usepackage{graphicx}
\usepackage{subfigure}
\usepackage{color}
\usepackage{epsfig}
\usepackage{epstopdf}
\usepackage{amsmath,amsthm,amssymb,amsfonts}
\usepackage{multicol}
\usepackage[]{caption2}
\usepackage{slashbox}
\usepackage{pict2e}

\newenvironment{list2}{
  \begin{list}{$\bullet$}{%
      \setlength{\itemsep}{0in}
      \setlength{\parsep}{0in} \setlength{\parskip}{0in}
      \setlength{\topsep}{0in} \setlength{\partopsep}{0in}
      \setlength{\leftmargin}{0.2in}}}{\end{list}}

\newtheorem{myth}{\bf Theorem}
\newtheorem{myprop}[myth]{\it Proposition}

\newtheorem{mylemma}[myth]{\bf Lemma}

\begin{document}

% paper title
% can use linebreaks \\ within to get better formatting as desire
\title{Joint Transceiver and Offset Design for Visible Light Communications with Input-dependent Shot Noise}

%\author{Qian Gao, Chen Gong and Zhengyuan Xu \footnote{The authors are with Key Laboratory of Wireless-Optical Communications, Chinese Academy of Sciences, University of Science and Technology of China, Hefei, China. Z. Xu is also with Shenzhen Graduate School, Tsinghua University, Shenzhen, China. Email: \{qgao,cgong821,xuzy\}@ustc.edu.cn.}}

\author{Qian Gao, \IEEEmembership{Member,~IEEE},~Chen Gong, \IEEEmembership{Member,~IEEE} and Zhengyuan Xu, \IEEEmembership{Senior Member,~IEEE}
\IEEEcompsocitemizethanks{\IEEEcompsocthanksitem This work was submitted to the Transaction on Wireless Communications on Feb. 16, 2016 and is currently under review. An abridged version of this manuscript was accepted by the IEEE Globecom 2016.}}

\maketitle

\IEEEpeerreviewmaketitle

\vspace{-0.2in}
\begin{abstract}
In this paper, we investigate the problem of the joint transceiver and offset design (JTOD) for point-to-point multiple-input-multiple-output (MIMO) and multiple user multiple-input-single-output (MU-MISO) visible light communication (VLC) systems. Both uplink and downlink multi-user scenarios are considered. The shot noise induced by the incoming signals is considered, leading to a more realistic MIMO VLC channel model. Under key lighting constraints, we formulate non-convex optimization problems aiming at minimizing the sum mean squared error. To optimize the transceiver and the offset jointly, a gradient projection based procedure is resorted to. When only imperfect channel state information is available, a semidefinite programming (SDP) based scheme is proposed to obtain robust transceiver and offset. The proposed method is shown to non-trivially outperform the conventional scaled zero forcing (ZF) and singular value decomposition (SVD) based equalization methods. The robust scheme works particularly well when the signal is much stronger than the noise.
%Change=R.3.Minor.1

\end{abstract}


\section{Introduction}
The recent decade has witnessed the visible light communication (VLC) adopting light-emitting diodes (LEDs) as a competent complement for radio frequency communications (RFC) in both indoor and outdoor environments \cite{Elgala}. At the same time, advances in LED manufacturing have prepared the landing of various of VLC products on market \cite{Kim}. At the earlier stage, the use of blue LEDs with a yellow phosphor coating dominated due to low cost and complexity concerns. Later, red/green/blue (RGB) LEDs received more attention, especially when a higher rate is required or where designated color other than white illumination is necessary and achieved through adjusting the relative average intensities of the colored LEDs \cite{Drost10,Bai12,IEEE11}.

Besides wider spectrum, LED-based VLC has another inherent advantage over the RFC when utilized inside small cells; that is the limited illumination coverage that prevents excessive inter-cell interference. Thus it is a viable option to provide the ``last few meters'' access for the next generation wireless communication network frequently characterized by very high throughput per unit area. To achieve a very high rate, multiplexing gain offered by optical multiple-input-multiple-output (MIMO) has been explored \cite{Zeng09,Biagi2,Nuwanpriya}.

Inside a small-cell with MIMO multi-color VLC, it is typical that the channel correlations (ChC) are high, unless reduced intentionally, for example by resorting to angle diversity receivers or imaging receivers \cite{Nuwanpriya}. With RGB LEDs, color cross-talks (CoC) may exist given imperfect color filters, allowing spectrum leakage from a neighboring band(s), thus further deteriorates system detection performance. We take into account both ChC and CoC in this paper, characterized by two non-diagonal matrices, the Kronecker product of which gives the overall channel.

Transceiver design for a MIMO VLC channel is an emerging topic, where literature has schemes for point-to-point scenarios \cite{Kai,Park,Qian1,Qian2} and for multi-user scenarios  \cite{Hao,Chaaban,Pearce,Karagiannidis}.  Capacity analysis and signal processing schemes for VLC with signal-dependent noise have attracted increasing interest recently \cite{JiangzhouWang,Pergoloni,Moser,Chaaban2,Lapidoth}. A cost-dependent model is given in \cite{Chaaban2}, which models the input-dependent noise as an average intensity-dependent noise. As a step upon previous works, we consider a joint transceiver and offset design for both point-to-point and multi-user MIMO VLC, with ChC and CoC considered. Further, input-dependent shot noise is also taken into account with the joint design.
An important observation from our study is that the part of mean squared error (MSE)
caused by the shot noise is dependent on the DC offset but not signals when an M-ary pulse-amplitude modulation (PAM) with a zero mean is adopted.
In fact, the additional term in MSE caused by the shot noise depends only on the offset value but not the signal, as the latter is averaged out. Different from the existing methods, where the offset is arbitrarily chosen, it is optimized jointly with the transceiver in this paper.
% End of Change

Specifically, we formulate the joint transceiver and offset design (JTOD) problem under key lighting constraints, including non-negative intensity, arbitrary illumination color, dimming level and total optical power. The resulting optimization problem is highly non-convex. We show that the optimal post-equalizer $\mathbf{G}$ is of the Wiener filter form,
and the precoder $\mathbf{P}$ and the offset $\mathbf{b}$ can be solved by alternating convex optimizations.

Throughout the paper, we will use the following conventions. Boldface upper-case letters denote matrices, boldface lower-case letters denote column vectors, and standard lower-case letters denote scalars. By $(\cdot)^T$ and $(\cdot)^{-1}$ we denote the transpose and inverse operators, $\mathbb{E}(\cdot)$ the expectation operator. By $||\cdot||$, $||\cdot||_F$ and $||\cdot||_1$ we denote the Euclidean norm, the F-norm and the 1-norm. By $\otimes$ we denote the Kronecker product, $abs(\cdot)$ the element-wise absolute value operator, $\mathcal{D}(\cdot)$ the Jacobian operator, $\partial$ the partial derivative and ${\it tr}(\cdot)$ the trace operator. By $\mathbf{I}$, we denote the identity matrix and $\mathbf{e}_i$ the vector with all zeros but one at the $i$-th element.

\section{MIMO Multi-color VLC Channel}\label{sec2}

A MIMO multi-color VLC system can be very different from an RFC counterpart, especially when the lower-cost intensity modulation direct detection (IM/DD) transceivers are employed. Besides the positive and real intensity requirement, several other significant differences deserve attention from system designers.

\subsection{Signal-dependent Shot Noise}\label{sec2c}
% Change=R2.2
Conventionally, it is assumed that the VLC systems suffer only from thermal noise, i.e.  $n_{th}\sim\mathcal{N}(0,\sigma^2)$. However, the practical VLC systems offer very high signal-to-noise ratio, under which scenario the noise depends on the signal itself due to the random nature of photon emission in the LED \cite{Moser}. The distribution of input-dependent shot noise is $n_{sh}\sim\mathcal{N}(0,x\varsigma^2\sigma^2)$, where $x$ is the transmitted signal and $\varsigma^2$ is the scaling factor of shot noise variance, the range of which can be chosen according to \cite{Moser}. With the additional shot noise term, there is no proof that the optimality of existing transceiver structure in literature still holds. Therefore, it is important to investigate the impact of the shot noise on system design optimality.
% End of Change=R2.2
\subsection{Channel Correlation}\label{sec2a}
Compared with communication networks employing RF transceivers, VLC networks in many scenarios face a more severe interference problem among adjacent cells/LEDs, although if taken good care of this can be a benefit when local coverage is required. The channel correlations can be too high to yield intended multiplexing gain, especially when transmitters or receivers are placed close to each other. ChC-reduction techniques are needed, e.g. by adopting angle diversity receivers or imaging detectors. With four interfering white LED and four detectors, sample VLC channel $\bar{\mathbf{H}}$ can be found in \cite{Park}.

\subsection{Color Cross-talks}\label{sec2b}
The color diversity is unique with VLC, but imperfect color filtering may result in CoC, which can be described by the following matrix

\begin{align}
\tilde{\mathbf{H}}=
&\begin{bmatrix}
1-\xi & \xi & 0\\
\xi & 1-2\xi & \xi\\
0 & \xi & 1-\xi \\
\end{bmatrix},
\end{align}

The overall channel is the Kronecker product of the two channels, i.e.
\begin{equation}
\mathbf{H}=\tilde{\mathbf{H}}\otimes \bar{\mathbf{H}}
\end{equation}
This channel model incorporates both ChC and CoC and will be adopted in our transceiver design. Advanced modulation schemes are available taking advantage of the color diversity, and we refer the readers to \cite{Butala,Pergoloni}.
\subsection{Further Assumptions}

In this paper, we assume a full channel status information (CSI) knowledge at both the transmitter and the receiver ends. Compare to the RFC, the channel status of VLC is much more stable, mostly with no fading associated \cite{Yu}. Thus, less frequent CSI update is required, leading to non-trivially reduced complexities for the CSI-dependent signal processing algorithms, including the transceiver design problem considered in this paper. Our design is firstly carried out assuming no channel error for the CSI, and then we provide an SDP-based algorithm to deal with an imperfect CSI case.

For the systems considered in this paper, the optimizations of $\mathbf{P}$, $\mathbf{G}$ and $\mathbf{b}$ can be done either locally at both the transmitter and receiver, or only at the transmitter, who then feedforward $\mathbf{G}$ and $\mathbf{b}$ to the receiver. If the receiver has high computation ability, local optimization is preferred since error that may be caused by the feedforward process is avoided. On the network level, we have left some practical issues to further studies, such as data synchronization, a handover procedure, etc.

\section{MIMO Point-to-Point VLC}\label{sec3}

\subsection{System Model}\label{sec3a}

\begin{figure}[htbp]
\centering
\centerline{\includegraphics[width=1.0\columnwidth]{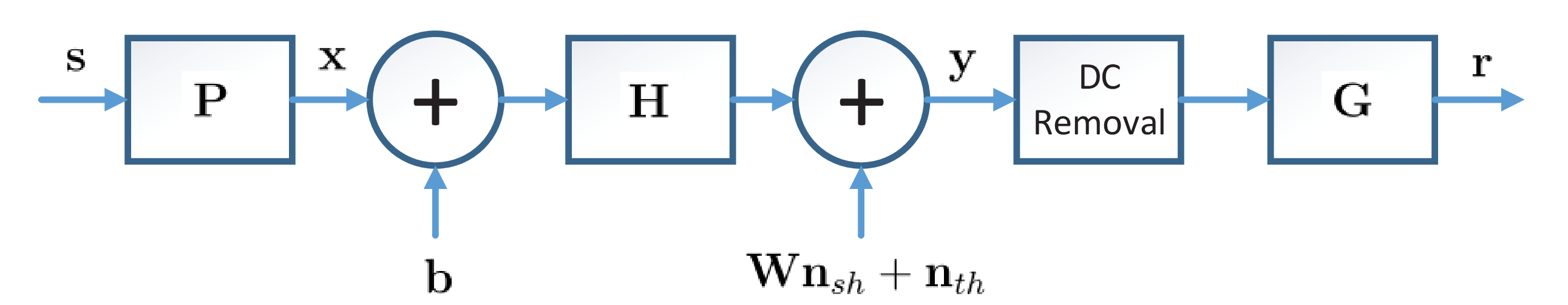}}
\caption{MIMO point-to-point VLC system block diagram.}
\end{figure}

\newcounter{FandQ}
\begin{figure*}[hb]
\hrulefill \setcounter{FandQ}{\value{equation}}
\setcounter{equation}{10}
\begin{align}
\text{MSE}&=\mathbb{E}_{\mathbf{s},\mathbf{n}_{sh},\mathbf{n}_{th}}
\big\{\parallel(\mathbf{G}\mathbf{H}\mathbf{P}-\mathbf{I})\mathbf{s}+\mathbf{G}
\mathbf{W}\mathbf{n}_{sh}+\mathbf{G}\mathbf{n}_{th}\parallel^2\big\}\notag\\
&=\mathbb{E}\big\{tr[(\mathbf{G}\mathbf{H}\mathbf{P}-\mathbf{I})\mathbf{s}+\mathbf{G}
\mathbf{W}\mathbf{n}_{sh}+\mathbf{G}\mathbf{n}_{th})^T(\mathbf{G}\mathbf{H}\mathbf{P}-\mathbf{I})\mathbf{s}+\mathbf{G}
\mathbf{W}\mathbf{n}_{sh}+\mathbf{G}\mathbf{n}_{th})] \big\}\notag\\
&=tr[(\mathbf{G}\mathbf{H}\mathbf{P}-\mathbf{I})\mathbf{R}_s(\mathbf{G}\mathbf{H}\mathbf{P}-\mathbf{I})^T]
+\sigma^2tr(\mathbf{G}\mathbf{G}^T)+\varsigma^2\sigma^2tr[(\mathbf{G}^T\mathbf{G}\mathbb{E}_{\mathbf{s}}(\mathbf{W}\mathbf{W}^T))]\notag\\
\end{align}
\begin{align}
&\mathbb{E}_{\mathbf{s}}(\mathbf{W}\mathbf{W}^T)
=\mathbb{E}_{\mathbf{s}}(\text{diag}(\mathbf{H}\mathbf{P}\mathbf{s}+\mathbf{H}\mathbf{b}))\notag\\
&=\begin{bmatrix}
\mathbb{E}_{\mathbf{s}}(\mathbf{e}_1^T\mathbf{H}\mathbf{P}\mathbf{s}+\mathbf{e}^T_1\mathbf{H}\mathbf{b}) & 0 & \ldots & 0\\
0 & \mathbb{E}_{\mathbf{s}}(\mathbf{e}^T_2\mathbf{H}\mathbf{P}\mathbf{s}+\mathbf{e}^T_2\mathbf{H}\mathbf{b})  & \ldots & 0\\
\vdots & \vdots & \ddots & \vdots\\
0  & 0 & \ldots & \mathbb{E}_{\mathbf{s}}(\mathbf{e}^T_{3N_t}\mathbf{H}\mathbf{P}\mathbf{s}+\mathbf{e}^T_{3N_t}\mathbf{H}\mathbf{b})
\end{bmatrix}=
&\begin{bmatrix}
\mathbf{e}^T_1\mathbf{H}\mathbf{b} & 0 & \ldots & 0\\
0 & \mathbf{e}^T_2\mathbf{H}\mathbf{b}  & \ldots & 0\\
\vdots & \vdots & \ddots & \vdots\\
0  & 0 & \ldots & \mathbf{e}^T_{3N_t}\mathbf{H}\mathbf{b}
\end{bmatrix}\label{12}
\end{align}
\setcounter{equation}{\value{FandQ}}
\end{figure*}

We first consider an indoor MIMO point-to-point VLC system. The transmitter is equipped with $N_t$ RGB LEDs and the receiver with $N_r$ photo-diode (PD) detectors, each with three color filters. The system diagram is shown in Fig. 1. A total of $K\leq \min\{3N_t,3N_r\}$ streams of data $\mathbf{s}(t)$ are input to a precoder $\mathbf{P}$. Each component of $\mathbf{s}(t)$ takes a value from independent M-PAM constellations,
% Change=R.3.Minor.3
which is symmetric on zero.
Since only non-negative intensity values are allowed by the LEDs, a DC offset $\mathbf{b}$ is added after the precoder. The time index is dropped in the following context as we are interested in a single time slot, and the intensity vector modulating the LEDs is written as
\begin{equation}
\mathbf{x}=\gamma(\mathbf{P}\mathbf{s}+\mathbf{b}),
\end{equation}
where $\gamma$ denotes the electro-to-opto conversion factor.
% Change=R2.3
The amplitudes of elements in $\mathbf{s}$ are in the range $s_{i}\in [-d,d],~i\in\{1,\ldots,K\}$ and each take a value from the same M-PAM constellation. Assuming independency of each data stream, the covariance matrix $\mathbf{s}$ is diagonal with identical elements, i.e.
\begin{equation}
\mathbf{R}_s=r\mathbf{I},
\end{equation}
\begin{equation}
r=\frac{d^2}{3}\bigg(\frac{M+1}{M-1}\bigg).
\end{equation}
% End of Change=R2.3

The received signal after the color filters is
\begin{align}
\mathbf{y}&=\eta\gamma\mathbf{H}\mathbf{x}+\mathbf{W}(\mathbf{P},\mathbf{b})\mathbf{n}_{sh}+\mathbf{n}_{th}\notag\\
&=\eta\gamma\mathbf{H}\mathbf{P}\mathbf{s}+\eta\gamma\mathbf{H}\mathbf{b}+\mathbf{W}(\mathbf{P},\mathbf{b})\mathbf{n}_{sh}+\mathbf{n}_{th},
\end{align}
where $\eta$ is the opto-to-electro conversion factor, $\mathbf{W}(\mathbf{P},\mathbf{b})\mathbf{n}_{sh}$ are the input-dependent shot noise terms. The variables $\mathbf{s}$, $\mathbf{n}_{sh}$ and $\mathbf{n}_{th}$ are independent. Each component of the shot noises $n_{sh,i}\sim \mathcal{N}(0,\varsigma^2\sigma^2)$ and the $W_{i,i}(\mathbf{P},\mathbf{b})$ is the $i$-th component of the diagonal matrix $\mathbf{W}$
\begin{equation}
W_{i,i}(\mathbf{P},\mathbf{b})=\sqrt{\eta\gamma\mathbf{e}^T_i(\mathbf{H}\mathbf{P}\mathbf{s}+\mathbf{H}\mathbf{b})}.
\end{equation}
We write $\mathbf{W}$ instead of $\mathbf{W}(\mathbf{P},\mathbf{b})$ for notational simplicity in the following context. Also, to facilitate analysis we assume $\eta=\gamma=1$, and the channel model is simplified as
\begin{align}
\mathbf{y}=\mathbf{H}\mathbf{P}\mathbf{s}+\mathbf{H}\mathbf{b}+\mathbf{W}\mathbf{n}_{sh}+\mathbf{n}_{th}.
\end{align}

Therefore, with receiver side CSI knowledge the deterministic terms $\mathbf{H}\mathbf{b}$ can be cancelled from $\mathbf{y}$ to obtain
\begin{align}
\bar{\mathbf{y}}=\mathbf{H}\mathbf{P}\mathbf{s}+\mathbf{W}\mathbf{n}_{sh}+\mathbf{n}_{th}.
\end{align}
A minimum mean square error based post-equalizer $\mathbf{G}$ is applied \cite{Tao}, such that the recovered symbol

\begin{align}
\hat{\mathbf{s}}&=\arg\min_{\mathbf{s}\in \mathcal{S}}\mathbb{E}\bigg(\parallel \mathbf{r}-\mathbf{s}\parallel^2\bigg)\notag\\
&=\arg\min_{\mathbf{s}\in\mathcal{S}}
\mathbb{E}\bigg(\parallel(\mathbf{G}\mathbf{H}\mathbf{P}-\mathbf{I})\mathbf{s}+\mathbf{G}
\mathbf{W}\mathbf{n}_{sh}+\mathbf{G}\mathbf{n}_{th}\parallel^2\bigg).
\end{align}

The $\text{MSE}$ is calculated by plugging (12) into (11) at the bottom of this page, i.e.

\setcounter{equation}{12}

\begin{align}
&\text{MSE}=
3N_tr+r{\it tr}(\mathbf{H}\mathbf{P}\mathbf{P}^T\mathbf{H}^T\mathbf{G}^T\mathbf{G})-r{\it tr}(\mathbf{G}\mathbf{H}\mathbf{P})\notag\\
&-r{\it tr}(\mathbf{P}^T\mathbf{H}^T\mathbf{G}^T)+\sigma^2tr(\mathbf{G}^T\mathbf{G})
+\varsigma^2\sigma^2{\it tr}\{\text{diag}(\mathbf{H}\mathbf{b})\mathbf{G}^T\mathbf{G}\}.\label{13}
\end{align}
It is seen that the last term $\varsigma^2\sigma^2{\it tr}\{\text{diag}(\mathbf{H}\mathbf{b})\mathbf{G}^T\mathbf{G}\}$ caused by the shot noise is offset dependent only, while the signal dependent part is averaged out.

\subsection{Key Lighting Constraints}\label{sec3b}

\subsubsection{Non-negative Intensity}
The intensity vector $\mathbf{x}$ modulating the LEDs has to take non-negative values, i.e.
\footnote{A vector $\mathbf{x}\leq(\geq) 0$ means that it is element-wise non-positive (non-negative).}
\begin{equation}
\mathbf{x}=\mathbf{P}\mathbf{s}+\mathbf{b}\geq 0.
\end{equation}
However, this constraint is signal-dependent, thus we resort to the following sufficient condition
\begin{equation}
abs(\mathbf{P})\boldsymbol{\delta}-\mathbf{b}\leq 0,
\end{equation}
where $abs(\mathbf{P})$ means element-wise absolute value and $\boldsymbol{\delta}=\delta\mathbf{1}$.

\subsubsection{Dimming Control and Total Optical Power Constraint}

The dimming control and total optical power constraint can be written as a single equation
% Change=R3.Minor.12
\begin{equation}
\mathbf{1}^T\mathbf{b}=\beta P_T,\label{14}
\end{equation}
where $\beta\in(0,1]$ is the dimming level, $P_T$ is the maximally allowed total optical power. They both rely on the offset only. The \eqref{14} has assumed unitary electro-to-opto conversion factor, and it is also the average power constraint since the oscillating part is averaged out, i.e.
\begin{align}
\mathbb{E}(\mathbf{P}\mathbf{s}+\mathbf{b})=\mathbf{P}\mathbb{E}(\mathbf{s})+\mathbf{b}
=\mathbf{b}.
\end{align}

%%%%%%%%%%%%%%%%%%%%%%%%%%%%%%%%%%%%%%%%%%%%%%%%%%%%%%%%%%%%%%%%%%%%%%%%%%%%%%
\newcounter{FandQ4}
\begin{figure*}[ht]
\setcounter{FandQ4}{\value{equation}}
\setcounter{equation}{23}
\begin{equation}
\begin{aligned}
& \underset{\mathbf{P},\mathbf{b}}{\text{min}}
& & {{\it tr}\big(\mathbf{H}\mathbf{P}\mathbf{P}^T\mathbf{H}^T\mathbf{G}^{(0),T}\mathbf{G}^{(0)}-2\mathbf{G}^{(0),T}\mathbf{H}\mathbf{P}
+\frac{\sigma^2}{r}\mathbf{G}^{(0),T}\mathbf{G}^{(0)}+\frac{\varsigma^2\sigma^2}{r}\text{diag}(\mathbf{H}\mathbf{b})\mathbf{G}^{(0),T}\mathbf{G}^{(0)}}\big) \\
& \text{s.t.}
&& abs(\mathbf{P})\boldsymbol{\delta}-\mathbf{b}\leq 0\\
&&& \boldsymbol{\Pi}\mathbf{b}=\beta P_T\bar{\mathbf{b}}
\end{aligned}
\end{equation}
\hrulefill \setcounter{equation}{\value{FandQ4}}
\end{figure*}

\newcounter{FandQ6}
\begin{figure*}[hb]
\hrulefill \setcounter{FandQ6}{\value{equation}}
\setcounter{equation}{28}
\begin{align}
\mathcal{D}_{\text{vec}(\mathbf{P})}
\text{MSE}=vec^T\bigg(\mathbf{H}^T\big(\varsigma^2\sigma^2
\text{diag}(\mathbf{H}\mathbf{b})+\sigma^2\mathbf{I}\big)^{-1}\mathbf{H}\mathbf{P}\mathbf{Y}^{-2}\bigg)+vec^T\bigg(\mathbf{Y}^{-2}\mathbf{P}^T\mathbf{H}^T\big(\varsigma^2\sigma^2
\text{diag}(\mathbf{H}\mathbf{b})+\sigma^2\mathbf{I}\big)^{-1}\mathbf{H}\bigg)\mathbf{\Pi},\label{25}
\end{align}
where $\boldsymbol{\Pi}$ is a permutation matrix such that $\text{d}\text{vec}(\mathbf{P}^T)=\boldsymbol{\Pi}\text{d}\text{vec}(\mathbf{P}).$ And
\begin{equation}
\mathbf{Y}=r^{-1}\mathbf{I}+\mathbf{P}^T\mathbf{H}^T
\big(\varsigma^2\sigma^2\text{diag}(\mathbf{H}\mathbf{b})+\sigma^2\mathbf{I}\big)^{-1}\mathbf{H}\mathbf{P}.
\end{equation}
\begin{equation}
\mathcal{D}_{\mathbf{b}}\text{MSE}=vec^T\big[(\varsigma^2\sigma^2\text{diag}(\mathbf{H}\mathbf{b})+\sigma^2\mathbf{I})^{-1}
\mathbf{H}\mathbf{P}\mathbf{Y}^{-2}\mathbf{P}^T\mathbf{H}^T(\varsigma^2\sigma^2\text{diag}(\mathbf{H}\mathbf{b})+
\sigma^2\mathbf{I})^{-1}\big]bvec(\mathbf{\Phi}_b),\label{26}
\end{equation}
\setcounter{equation}{\value{FandQ6}}
\end{figure*}

%%%%%%%%%%%%%%%%%%%%%%%%%%%%%%%%%%%%%%%%%%%%%%%%%%%%%%%%%%%%%%%%%%%%%%%%%%%%%%

\subsubsection{Illumination Color Constraint}
System designers can determine the illumination color by adjusting the so termed color ratio vector $\bar{\mathbf{b}}$ on the CIE color space
\begin{equation}
\bar{\mathbf{b}}=[b_r~ b_g~ b_b]^T\geq \textbf{0},
\end{equation}
\begin{equation}
b_r+b_g+b_b=1,
\end{equation}
where the subscripts $b_r$, $b_g$, and $b_b$ stand for the relative intensity of the red, green and blue LEDs respectively, i.e. $\bar{\mathbf{b}}$ is tristimulus. The illumination color constraint can thus be written as
\begin{align}
\boldsymbol{\Pi}\mathbf{b}=\beta P_T\bar{\mathbf{b}},\label{16}
\end{align}
where matrix $\boldsymbol{\Pi}$ is of size $3\times 3N_t$ and contains zeros and ones only, which serve to sum up individual color intensities. For a two RGB LEDs case, the matrix
%$\{R1,R2,G1,G2,B1,B2\}$

\begin{equation}
\boldsymbol{\Pi}=
\begin{bmatrix}
1 & 1 & 0 & 0 & 0 & 0\\
0 & 0 & 1 & 1 & 0 & 0\\
0 & 0 & 0 & 0 & 1 & 1
\end{bmatrix}.
\end{equation}
It is straightforward to show that \eqref{14} is absorbed into \eqref{16}. Thus, one equality includes both the optical power and illumination color requirements.

\subsection{The MSE Minimization Problem}

The objective of this design is to minimize the MSE in \eqref{13} by optimizing the post-equalizer, the precoder and the offset, subject to the lighting constraints. The optimization problem is formulated as

\setcounter{equation}{21}

\begin{equation}
\begin{aligned}
& \underset{\mathbf{G},\mathbf{P},\mathbf{b}}{\text{min}}
& & {{\text{MSE}}}\\
& \text{s.t.}
&& abs(\mathbf{P})\boldsymbol{\delta}-\mathbf{b}\leq 0\\
&&& \boldsymbol{\Pi}\mathbf{b}=\beta P_T\bar{\mathbf{b}}.\label{21}
\end{aligned}
\end{equation}
This problem is jointly non-convex in $(\mathbf{G},\mathbf{P},\mathbf{b})$.

Other reasonable objectives besides the sum MSE include minimization of the arithmetic mean of the MSEs, the geometric mean of the MSEs, the determinant of the MSEs and the maximum of the MSEs. We refer the readers interested to a further reading of the paper \cite{Palomar}. The joint design problem may be studied with other objectives such as maximizing the available system achievable rates in literature \cite{Lapidoth,Moser,Chaaban2}.

\subsection{The Proposed Optimization Method}

As the post-equalizer $\mathbf{G}$ only appears in the objective function of \eqref{21}, the optimum is obtainable by setting the associated gradient of MSE to zero, i.e.
\begin{equation}
\mathcal{D}_{\mathbf{G}}\text{MSE}=\mathbf{0}.\label{22}
\end{equation}
The solution is
\begin{align}
\mathbf{G}=r\mathbf{P}^{T}\mathbf{H}^T\big(r\mathbf{H}\mathbf{P}\mathbf{P}^{T}\mathbf{H}^T
+\sigma^2\varsigma^2\text{diag}(\mathbf{H}\mathbf{b})+\sigma^2\mathbf{I}\big)^{-1},\label{23}
\end{align}
which is of the Wiener filter form. Plug \eqref{23} into \eqref{13} the following expression of MSE is obtained

\setcounter{equation}{24}

\begin{equation}
\text{MSE}=tr\bigg(r^{-1}\mathbf{I}+\mathbf{P}^T\mathbf{H}^T
\big(\varsigma^2\sigma^2\text{diag}(\mathbf{H}\mathbf{b})+\sigma^2\mathbf{I}\big)^{-1}
\mathbf{H}\mathbf{P}\bigg)^{-1}\label{25}
\end{equation}
The resulting optimization problem is
\begin{equation}
\begin{aligned}
& \underset{\mathbf{P},\mathbf{b}}{\text{min}}
& & {\text{MSE}} \\
& \text{s.t.}
&& abs(\mathbf{P})\boldsymbol{\delta}-\mathbf{b}\leq 0\\
&&& \boldsymbol{\Pi}\mathbf{b}=\beta P_T\bar{\mathbf{b}}.\label{22}
\end{aligned}
\end{equation}
\begin{myprop}
The objective function in \eqref{22} is not convex in $\mathbf{P}$ or $\mathbf{b}$.
\end{myprop}

\begin{proof}
See Appendix A.
\end{proof}

Therefore, we propose to solve $\mathbf{P}$ and $\mathbf{b}$ one after another.
The following relationship between the differential and the Jacobian on $\mathbf{P}$ is useful
\begin{equation}
\text{d}\text{MSE}=\mathcal{D}_{\text{vec}(\mathbf{P})}
\text{MSE}\cdot\text{d}\text{vec}(\mathbf{P}).
\end{equation}
The Jacobian of the objective $\text{MSE}$  w.r.t $\text{vec}(\mathbf{P})$ is derived as in (29) at the bottom of this page, where $\text{vec}(\cdot)$ is the vectorization operator.
The Jacobian of the objective $\text{MSE}$  w.r.t $\mathbf{b}$ is derived as in (31).

\setcounter{equation}{27}

The $3N_t\times (3N_t)^2$ matrix $\mathbf{\Phi}_b$ is of the following form
\begin{equation}
\mathbf{\Phi}_b
=
\begin{bmatrix}
\mathbf{e}_1^T\mathbf{H} & \mathbf{0}^T & \ldots & \mathbf{0}^T\\
\mathbf{0}^T & \mathbf{e}_2^T\mathbf{H} & \ldots & \mathbf{0}^T\\
\vdots & \vdots & \ddots & \vdots\\
\mathbf{0}^T  & \mathbf{0}^T & \ldots & \mathbf{e}_{3N_t}^T\mathbf{H}
\end{bmatrix},
\end{equation}
where we define $\text{bvec}(\cdot)$ as a non-standard block vectorization operator to rearrange the wide matrix $\mathbf{\Phi}_b$  into a tall one with dimension $(3N_t)^2\times 3N_t$.
% Change=R3.Minor.7
To obtain an optimized pair of precoder and offset, a gradient projection procedure as in Algorithm 1 is applied \cite{Bertsekas}. The feasible region $\Omega_{\mathbf{P}}$ of variable $\mathbf{P}$ determined by the first constraint of problem \eqref{21} is convex and the feasible region $\Omega_{\mathbf{b}}$ of variable $\mathbf{b}$ determined by both constraints is also convex. Therefore, to project a vector to a convex set is to find the vector in the set that has the minimum distance to it.

\begin{algorithm}[ht]
        \caption{Iterative Gradient Projection Algorithm}
        \label{GPAlgorithm}           %new code
        \begin{algorithmic}[1]
            \renewcommand{\algorithmicrequire}{\textbf{Input:}}
            \renewcommand{\algorithmicensure}{\textbf{Output:}}
            \Require Initialize $\text{vec}(\mathbf{P}^{(0)})$ and $\mathbf{b}^{(0)}$, $k=0$;
            \Ensure  Converged $\text{vec}\big(\mathbf{P}^{*}\big)=\text{vec}(\mathbf{P}^{(k)})$ and $\mathbf{b}^{*}=\mathbf{b}^{(k)}$;
            %\DontPrintSemicolon
            \Repeat
            \State
             Compute the Jacobian matrix $\mathcal{D}_{\text{vec}(\mathbf{P})}\text{MSE}$ by \eqref{25};
            \State
            $\text{vec}^T(\widetilde{\mathbf{P}}^{(k)}) = \text{vec}^T(\mathbf{P}^{(k)})+\alpha^{(k)}\mathcal{D}_{\text{vec}(\mathbf{P})}f(\mathbf{P}^{(k)})$;
            \State
            $\text{vec}^T(\bar{\mathbf{P}}^{(k)})$ = projection of $\text{vec}^T(\widetilde{\mathbf{P}}^{(k)})$ onto the feasible region $\Omega_{\mathbf{P}}$;
            \State
            $\text{vec}^T(\mathbf{P}^{(k+1)})=\text{vec}^T(\bar{\mathbf{P}}^{(k)}) + \gamma^{(k)}(\text{vec}^T(\bar{\mathbf{P}}^{(k)}) -\text{vec}^T(\mathbf{P}^{(k)}) )$;
            \State
            Set $\mathbf{P}^k=\mathbf{P}^{k+1}$ and compute the Jacobian matrix $\mathcal{D}_{\mathbf{b}}\text{MSE}$ by \eqref{26};
            \State
            Update the vector $\mathbf{b}$ by a similar procedure with STEP 3$-$5;
            \State
            $k=k+1$;
            \Until $||\text{vec}^T(\mathbf{P}^{(k)})-\text{vec}^T(\mathbf{P}^{(k-1)})|| < \epsilon_P$ and $||\mathbf{b}^{(k)}-\mathbf{b}^{(k-1)}|| < \epsilon_b$;
        \end{algorithmic}
    \end{algorithm}
The stopping criterion parameters $\epsilon_P=\epsilon_b=10^{-4}$ are chosen, $\Omega_{\mathbf{P}}$ and $\Omega_{\mathbf{b}}$ are the convex feasible regions, and the step sizes $\alpha^{(k)}$ and $\gamma^{(k)}$ are calculated based on the Armijo rule. The time complexity of gradient operation is $\mathcal{O}(K^4)$, observed from (36)-(38), and the time complexity of projection operation is $\mathcal{O}\big((K^2)^{3.5}\big)$.

\subsection{The Minimal Dimming Level Problem}
There are certain cases that lower dimming level is preferred, provided that the communication performance is guaranteed. In these cases, the dimming factor $\beta$ is also an optimization variable besides the transceiver and offset. The associated optimization problem is formulated as follows
\setcounter{equation}{31}
\begin{equation}
\begin{aligned}
& \underset{\mathbf{G},\mathbf{P},\mathbf{b},\beta}{\text{min}}
& & {\beta} \\
& \text{s.t.}
&& \text{MSE}\leq \varepsilon\\
&&& abs(\mathbf{P})\boldsymbol{\delta}-\mathbf{b}\leq 0\\
&&& \boldsymbol{\Pi}\mathbf{b}=\beta P_T\bar{\mathbf{b}},\label{33}
\end{aligned}
\end{equation}
where the worst MSE performance is constrained with a small number $\varepsilon$. Problem \eqref{33} is not jointly convex in $(\mathbf{G},\mathbf{P},\mathbf{b},\beta)$, but the problem is convex in $(\mathbf{P},\mathbf{b},\beta)$ if $\mathbf{G}$ is assumed to be a Wiener filter.

\section{Multi-user MISO VLC}
\subsection{Downlink Transceiver Design with Perfect CSI}

A closely related problem to the point-to-point MIMO system is the multi-user multiple input single output (MU-MISO) broadcast system as shown in Figure 2. The receiver-side single antenna deployment is particularly suitable for VLC, considering the limited size of user handsets, e.g. mobile phones. Moreover, it is also reasonable to adopt the multiple (even massive number of) transmitters configurations on the other end, since LED arrays are often naturally used for illumination uniformity.
\begin{figure}[htbp]
\centerline{\includegraphics[width=1.0\columnwidth]{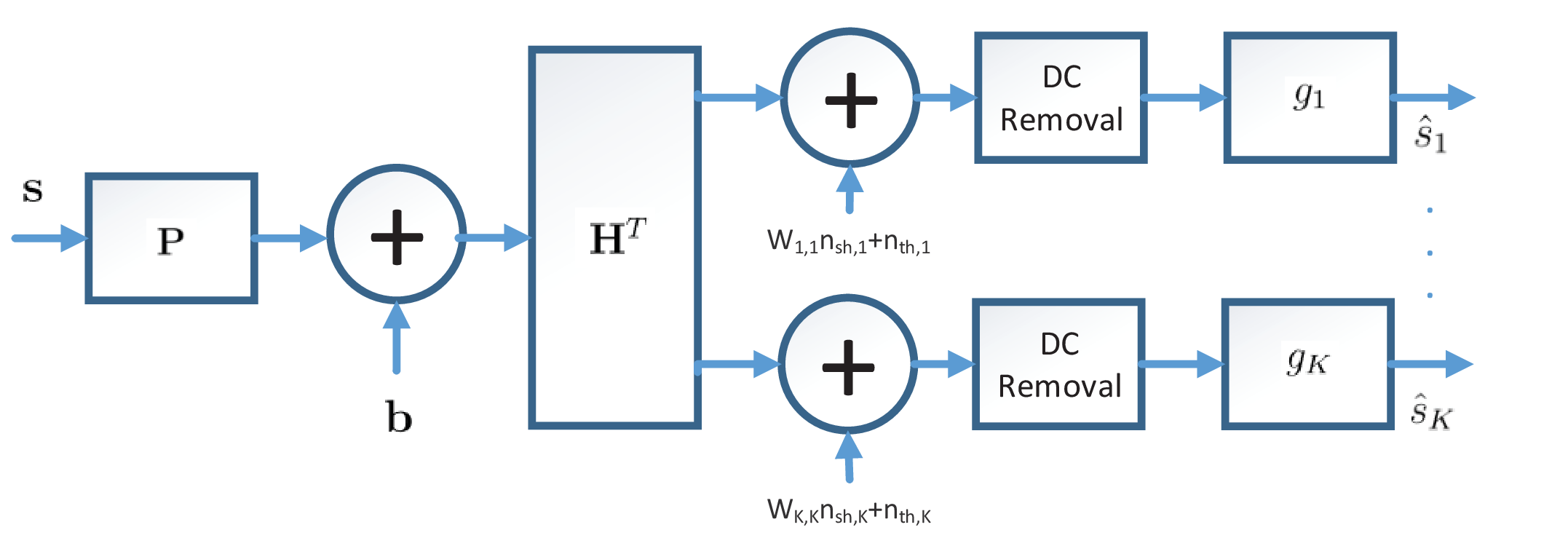}}
\vspace{-0.1in}
\caption{MU-MISO Downlink VLC system block diagram.}
\end{figure}
Assume that $K$ users are in presence in a room environment.
% Change=R3.Minor.9
The data for $k$-th user $s_k$ multiplies its beamforming vector $\mathbf{p}_k$ of size $N_t\times 1$, and the summed signal for all users is then added with an offset vector $\mathbf{b}$, thus the resultant signal is
\begin{equation}
\mathbf{x}=\sum_{k=1}^K\mathbf{p}_ks_k+\mathbf{b},
\end{equation}

The received intensity at user $k$ is
\begin{align}
&y_k=\mathbf{h}_k^T\mathbf{x}+(\mathbf{h}_k^T\mathbf{x})^{1/2}n_{sh,k}+n_{th,k}\notag\\
&=\mathbf{h}_k^T\sum_{j=1}^K\mathbf{p}_js_j+\mathbf{h}_k^T\mathbf{b}
+\big(\mathbf{h}_k^T\sum_{j=1}^K\mathbf{p}_js_j+\mathbf{h}_k^T\mathbf{b}\big)^{\frac{1}{2}}n_{sh,k}+n_{th,k},\label{41}
\end{align}
where $\mathbf{h}_k$ is the $k$-th column of $\mathbf{H}$. Then the second term in \eqref{41} is cancelled, resulting in
\begin{align}
&\bar{y}_k=\notag\\
&\mathbf{h}_k^T\mathbf{p}_ks_k+
\underbrace{\mathbf{h}_k^T\sum_{j\neq k}\mathbf{p}_js_j+\big(\mathbf{h}_k^T\sum_{j=1}^K\mathbf{p}_js_j
+\mathbf{h}_k^T\mathbf{b}\big)^{\frac{1}{2}}n_{sh,k}+n_{th,k}}_\text{interference plus noise}
\end{align}
Then a scalar post-equalizer $g_k$ multiplies each $\bar{y}_k$ to recover the $k$-th symbol
\begin{equation}
\hat{s}_k=g_k\bar{y}_k.
\end{equation}
The MSE of the $k$-th user is
\begin{align}
&\text{MSE}^{DL}_k=\mathbb{E}_{\mathbf{s},n_{sh,k},n_{th,k}}||\hat{s}_k-s_k||_2^2\notag\\
&=r\sum_{j=1}^Kg_k^2\mathbf{h}_k^T\mathbf{p}_j\mathbf{p}_j^T\mathbf{h}_k-2rg_k\mathbf{h}_k^T\mathbf{p}_k+\sigma^2g_k^2+r+\underbrace{\varsigma^2\sigma^2g^2_k\mathbf{h}_k^T
\mathbf{b}}_{\text{shot noise induced}}\notag\\\label{37}
%&=rg_k^2\sum_{j=1}^K\mathbf{p}_j^T\mathbf{H}_k\mathbf{p}_j-2rg_k\mathbf{h}_k^T\mathbf{p}_k+\sigma^2g_k^2+r+\varsigma^2\sigma^2g^2_k\mathbf{h}_k^T
%\mathbf{b}
\end{align}

By taking the gradient of \eqref{37}, the optimum $g_k$ is obtained as
\begin{equation}
g_k=r\mathbf{h}_k^T\mathbf{p}_k(r\sum_{j=1}^K\mathbf{h}_k^T\mathbf{p}_j\mathbf{p}_j^T\mathbf{h}_k+
\varsigma^2\sigma^2\mathbf{h}_k^T\mathbf{b}+\sigma^2)^{-1}.\label{38}
\end{equation}
By plugging \eqref{38} into \eqref{37} the following expression is obtained,
\begin{equation}
\text{MSE}_k^{DL}=r-r\mathbf{h}_k^T\mathbf{p}_k(\mathbf{h}_k^T\mathbf{P}\mathbf{P}^T\mathbf{h}_k+
\frac{\varsigma^2\sigma^2\mathbf{h}_k^T\mathbf{b}}{r}+\frac{\sigma^2}{r})^{-1}\mathbf{p}_k^T\mathbf{h}_k\label{39}
\end{equation}

The optimization problem minimizing the sum MSE is formulated as follows
\begin{equation}
\begin{aligned}
& \underset{\mathbf{P},\mathbf{b}}{\text{max}}
& & {\sum_{k=1}^Kr\mathbf{h}_k^T\mathbf{P}\mathbf{e}_k(\mathbf{h}_k^T\mathbf{P}\mathbf{P}^T\mathbf{h}_k+
\frac{\varsigma^2\sigma^2\mathbf{h}_k^T\mathbf{b}}{r}+\frac{\sigma^2}{r})^{-1}\mathbf{e}_k^T\mathbf{P}^T\mathbf{h}_k} \\
& \text{s.t.}
&& abs(\mathbf{P})\boldsymbol{\delta}-\mathbf{b}\leq 0\\
&&& \mathbf{1}^T\mathbf{b}=\beta P_T,
\end{aligned}
\end{equation}
where $\mathbf{p}_k=\mathbf{P}\mathbf{e}_k$. The problem can be solved by applying a similar gradient projection procedure as given in Algorithm 1.

\subsection{Downlink Transceiver Design with Imperfect CSI}
It is more practical to assume imperfect CSI at the base station. We model the imperfect CSI with a deterministic norm-bounded error model \cite{Jose}, i.e.
\begin{equation}
\mathbf{h}_k\in\{\hat{\mathbf{h}}_k+\boldsymbol{\delta}_k:||\boldsymbol{\delta}_k||\leq \rho_k\}\triangleq\mathcal{U}_k(\boldsymbol{\delta}_k),~\forall k\label{50}
\end{equation}
where $\hat{\mathbf{h}}_k$ and $\boldsymbol{\delta}_k$ are the $k$-th column of the estimated channel matrix and the corresponding estimation error vector, $\rho_k=s||\mathbf{h}_k||$ is the uncertain size of the $k$-th column of channel and scalar $s\in [0,1)$ is termed the uncertain scalar. The optimal post-equalizer is a Wiener filter as in \eqref{38} with $\mathbf{h}_k$ replaced with $\hat{\mathbf{h}}_k$. The joint design of transceiver and offset considering imperfect CSI for both the transmitter and receiver is included in Proposition 2, which is termed robust JTOD.
\begin{myprop}[]

The optimized precoder $\mathbf{P}^*$ and offset $\mathbf{b}^*$ can be found jointly through solving the following semidefinite programming (SDP), where channel uncertainty is assumed, and the maximum possible sum MSE is minimized:\\
\begin{equation}
\begin{aligned}
& \underset{\substack{\mathbf{P},\mathbf{b},\boldsymbol{\lambda},\boldsymbol{\tilde{\lambda}}\\ \boldsymbol{\nu},\boldsymbol{\tilde{\nu}},\varpi}}{\text{min}}
& & {\varpi} \\
& \text{s.t.}
&& \sum_{k=1}^K(\lambda_k+\tilde{\lambda}_k)+\lambda_0\leq \varpi ,\\
&&& ||\sigma\mathbf{g}||^2\leq\lambda_0\\
&&&
\begin{bmatrix}
    \lambda_k-\nu_k  & g_k\hat{\mathbf{h}}_k^T\mathbf{P}-\mathbf{e}_k^T &  \mathbf{0}^T \\
   (g_k\hat{\mathbf{h}}_k^T\mathbf{P}-\mathbf{e}_k^T)^T &   r^{-1}\mathbf{I} &  -\rho_kg_k\mathbf{P}^T \\
   \mathbf{0}  & -\rho_kg_k\mathbf{P} &  \nu_k\mathbf{I}
\end{bmatrix}\\
&&&\qquad \qquad \qquad \qquad \qquad \qquad \qquad \succeq 0,\forall k\\
&&&
\begin{bmatrix}
    \tilde{\lambda}_k-\varsigma^2\sigma^2g_k^2\hat{\mathbf{h}}_k^T\mathbf{b}-\tilde{\nu}_k(\frac{1}{2}\varsigma^2\sigma^2g_k^2)^2  & -\rho_k\mathbf{b}^T\\
   -\rho_k\mathbf{b} &  \tilde{\nu}_k\mathbf{I}
\end{bmatrix}\\
&&&\qquad \qquad \qquad \qquad \qquad \qquad \qquad\succeq 0,\forall k\\
%&&&
%\begin{bmatrix}
%    \lambda_k'-\nu_k'  & 0 &  \mathbf{0}^T \\
%   0 &   1 &  -\rho_k\mathbf{b}^T \\
%   \mathbf{0}  & -\rho_k\mathbf{b} &  \nu_k'\mathbf{I}
%\end{bmatrix}\succeq 0,~\forall k\\
&&& \nu_k\geq 0, \tilde{\nu}_k\geq 0,~\forall k\\
&&& abs(\mathbf{P})\boldsymbol{\delta}-\mathbf{b}\leq 0\\
&&& \mathbf{1}^T\mathbf{b}= \beta P_T,
\end{aligned}\label{51}
\end{equation}
where the square matrix $\mathbf{A}\succeq 0$ means $\mathbf{A}$ is positive semi-definite.
\end{myprop}
\begin{proof}
See Appendix B.
\end{proof}

\subsection{Uplink Transceiver Design and Comments on Duality}
\begin{figure}[htbp]
\centerline{\includegraphics[width=1.0\columnwidth]{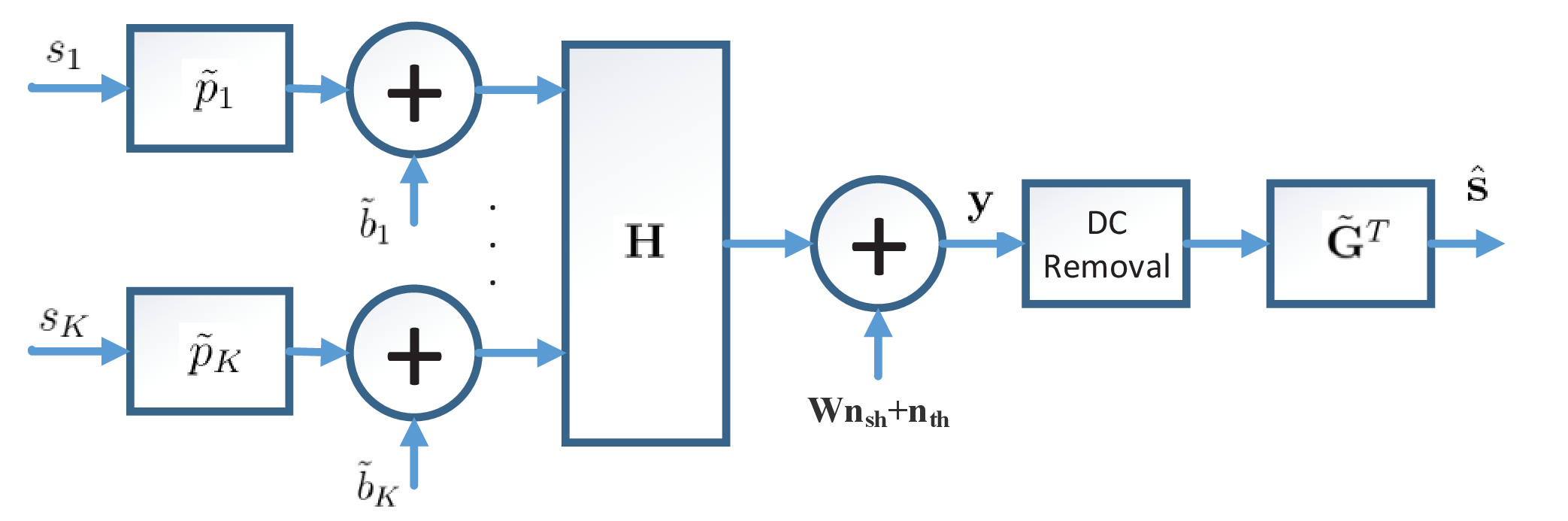}}
\vspace{-0.1in}
\caption{MU-MISO Uplink VLC system block diagram.}
\end{figure}

Consider the uplink multiple access channel in Figure 3, whose equivalence to the downlink counterpart is not (dis)proved for the visible light channel, although the duality relationship is well-known for radio frequency channel. The received intensity at the access point is
\begin{equation}
\mathbf{y}=\sum_{k=1}^K\mathbf{h}_k(\tilde{p}_ks_k+\tilde{b}_k)+\text{diag}^{1/2}\big(\sum_{k=1}^K
\mathbf{h}_k(\tilde{p}_ks_k+\tilde{b}_k)\big)\mathbf{n}_{sh}+\mathbf{n}_{th}.
\end{equation}
Removal of the constant term leads to
\begin{equation}
\bar{\mathbf{y}}=\sum_{k=1}^K\mathbf{h}_k\tilde{p}_ks_k+\text{diag}^{1/2}\big(\sum_{k=1}^K
\mathbf{h}_k(\tilde{p}_ks_k+\tilde{b}_k)\big)\mathbf{n}_{sh}+\mathbf{n}_{th}.
\end{equation}

The matrix $\tilde{\mathbf{G}}^T$, in this case, serves as the multiuser receiver, and the $k$-th recovered symbol is found in \eqref{46} and the MSE of the $k$-th link is found in \eqref{47}, both at the top of the next page.

\newcounter{FandQ7}
\begin{figure*}[ht]
 \setcounter{FandQ7}{\value{equation}}
\setcounter{equation}{44}
\begin{align}
\hat{s}_k=\tilde{\mathbf{g}}_k^T\mathbf{h}_k\tilde{p}_ks_k+\tilde{\mathbf{g}}_k^T\sum_{j\neq k}\mathbf{h}_j\tilde{p}_js_j+\tilde{\mathbf{g}}_k^T\text{diag}^{1/2}\big(\sum_{j=1}^K
\mathbf{h}_j(\tilde{p}_js_j+\tilde{b}_j)\big)\mathbf{n}_{sh}+\tilde{\mathbf{g}}_k^T\mathbf{n}_{th}.\label{46}
\end{align}
\begin{align}
\text{MSE}_k^{UL}&=\mathbb{E}_{\mathbf{s},\mathbf{n}_{sh},\mathbf{n}_{th}}||\hat{s}_k-s_k||_2^2=\sum_{j=1}^Kr\tilde{p}_k^2\tilde{\mathbf{g}}_k^T\mathbf{h}_j\mathbf{h}_j^T\tilde{\mathbf{g}}_k-2r\tilde{p}_k\tilde{\mathbf{g}}_k^T
\mathbf{h}_k+\sigma^2\tilde{\mathbf{g}}_k^T\tilde{\mathbf{g}}_k+r
+\underbrace{\varsigma^2\sigma^2\tilde{\mathbf{g}}_k^T\text{diag}(\sum_{j=1}^K\mathbf{h}_j\tilde{b}_j)\tilde{\mathbf{g}}_k}_{\text{shot noise induced term}}.\label{47}
\end{align}
\hrulefill \setcounter{equation}{\value{FandQ7}}
\end{figure*}

\setcounter{equation}{46}

The downlink uplink MSE duality is achieved if
\begin{equation}
\text{MSE}^{UL}_k=\text{MSE}^{DL}_k~~\forall k.\label{52}
\end{equation}

For the RF counterpart, the duality is studied on condition that electrical transmission power being the same. For VLC, duality requires the optical power consumption to equal for uplink and downlink channels, i.e.
\begin{equation}
\mathbf{1}^T\mathbf{b}=\mathbf{1}^T\mathbf{\tilde{b}},
\end{equation}
where $\mathbf{\tilde{b}}=[\tilde{b}_1~\tilde{b}_2~\ldots\tilde{b}_K]^T$. Besides, the non-negative constraints are unique to VLC and need to be guaranteed.
%\begin{equation}
%-abs(\mathbf{p}_k^T)\boldsymbol{\delta}+b_k\geq 0,~~\forall k
%\end{equation}
%\begin{equation}
%-abs(\tilde{p}_k)\delta +\tilde{b}_k\geq 0,~~\forall k.
%\end{equation}

It is known that without the shot noise induced term, the equation \eqref{52} is satisfied when
\begin{equation}
\mathbf{p}_k=\alpha_k\tilde{\mathbf{g}}_k~\text{and}~g_k=\alpha_k^{-1}\tilde{p}_k,\label{56}
\end{equation}
where $\alpha_k$ is the solution to the following equation \cite{Amine}
\begin{equation}
\mathbf{C}
\begin{bmatrix}
\alpha_1^2 \\
\vdots \\
\alpha_K^2
\end{bmatrix}=\sigma^2
\begin{bmatrix}
\tilde{p}_1^2 \\
\vdots \\
\tilde{p}_K^2
\end{bmatrix},
\end{equation}
where
\begin{equation}
\mathbf{C}_{k,j} =
\begin{cases}
\sum_{i\neq k}\tilde{p}_i^2\tilde{\mathbf{g}}_k^T\mathbf{h}_i\mathbf{h}_i^T\tilde{\mathbf{g}}_k+\sigma^2\tilde{\mathbf{g}}_k^T\tilde{\mathbf{g}}_k & k=j\\
-\tilde{p}_k^2\tilde{\mathbf{g}}_j^T\mathbf{h}_k\mathbf{h}_k^T\tilde{\mathbf{g}}_j & k\neq j.
\end{cases}
\end{equation}

To verify whether duality holds, it is required that the shot noise induced downlink and uplink terms be compared with the condition in \eqref{56}. The shot noise induced term for downlink is given by
\begin{equation}
\text{MSE}_{k,sh}^{DL}=\varsigma^2\sigma^2\alpha_k^{-2}\tilde{p}_k^2\mathbf{h}_k^T\mathbf{b},
\end{equation}
while the shot noise induced term for uplink is given by
\begin{equation}
\text{MSE}_{k,sh}^{UL}=\varsigma^2\sigma^2\alpha_k^{-2}\mathbf{p}_k^T\text{diag}(\sum_{j=1}^K\mathbf{h}_j\tilde{b}_j)\mathbf{p}_k.
\end{equation}
There is no guarantee of the equality of these two terms. Thus, duality does not hold for VLC downlink and uplink with an existence of signal-dependent shot noise. Therefore, the two problems should be studied independently. The uplink problem minimizing the sum MSE is formulated as
\begin{equation}
\begin{aligned}
& \underset{\tilde{\mathbf{G}},\tilde{\mathbf{p}},\tilde{\mathbf{b}}}{\text{min}}
& & {\sum_{k=1}^K\text{MSE}^{UL}_k} \\
& \text{s.t.}
&&
abs(\tilde{p}_k)\delta-\tilde{b}_k\leq 0,~~\forall k\\
&&& \mathbf{1}^T\mathbf{\tilde{b}}=\beta P_T,\label{56}
\end{aligned}
\end{equation}
where $\tilde{\mathbf{G}}=[\tilde{\mathbf{g}}_1,\tilde{\mathbf{g}}_2,\ldots,\tilde{\mathbf{g}}_K]$. The optimal $\tilde{\mathbf{g}}_k$ is also a Wiener filter
\begin{equation}
\tilde{\mathbf{g}}_k=r\tilde{p}_k\bigg(\tilde{p}_k^2\mathbf{H}\mathbf{H}^T+\varsigma^2\sigma^2
\text{diag}\bigg(\sum_{j=1}^K\mathbf{h}_j\tilde{b}_j\bigg)+\sigma^2\mathbf{I}\bigg)^{-1}\mathbf{h}_k\label{57}
\end{equation}
By plugging \eqref{57} into \eqref{56}, $\mathbf{P}$ and $\mathbf{b}$ can be obtained by solving the resulting problem with a similar gradient projection procedure as given in Algorithm 1.

\section{Simulation Results}\label{sec5}
In this section, we compare the performance of the proposed methods with the conventional scaled zero-forcing (ZF) based and singular value decomposition (SVD) based pre-equalization methods. For the ZF-based method, we have
\begin{equation}
\mathbf{P}=\mu_1\mathbf{H}^{-1},~~\mathbf{G}=\frac{1}{\mu_1}\mathbf{I},~~\mathbf{b}=\frac{\beta P_T}{3N_t}\mathbf{1},
\end{equation}
% Change=R3.15
where
\begin{equation}
\mu_1=\min_k\frac{b_k}{[abs(\mathbf{H}^{-1})]_{(k,:)}\cdot\mathbf{1}d},
\end{equation}
where $[\mathbf{A}]_{(k,:)}$ denotes the $k$-th row of $\mathbf{A}$. Thus, $\mu_1$ is a positive number to guarantee the positivity of intensity on all LEDs. The offset $\mathbf{b}$ is a fixed, and all the LEDs share the DC power. The recovered symbols are
\begin{equation}
\hat{\mathbf{s}}=\mathbf{s}+\frac{1}{\mu_1}\mathbf{W}^{ZF}\mathbf{n}_{sh}+\frac{1}{\mu_1}\mathbf{n}_{th},
\end{equation}
where
\begin{equation}
W^{ZF}_{i,i}=\sqrt{\mathbf{e}_i^T(\mu_1\mathbf{s}+\mathbf{H}\mathbf{b})},~~\forall i.
\end{equation}
For the SVD-based method, we decompose the channel $\mathbf{H}=\mathbf{U}\boldsymbol{\Sigma}\mathbf{V}^T$ and set
\begin{equation}
\mathbf{P}=\mu_2\mathbf{V}\boldsymbol{\Sigma}^{-1},~~\mathbf{G}=\frac{1}{\mu_2}\mathbf{U}^T,~~\mathbf{b}=\frac{\beta P_T}{3N_t}\mathbf{1},
\end{equation}
where $\mu_2$ is calculated in a similar way to $\mu_1$. The recovered symbols are
\begin{equation}
\hat{\mathbf{s}}=\mathbf{s}+\frac{1}{\mu_2}\mathbf{U}^T\mathbf{W}^{SVD}\mathbf{n}_{sh}+\frac{1}{\mu_2}\mathbf{U}^T\mathbf{n}_{th},
\end{equation}
where
\begin{equation}
W^{SVD}_{i,i}=\sqrt{\mathbf{e}_i^T(\mu_2\mathbf{H}\mathbf{V}\boldsymbol{\Sigma}^{-1}\mathbf{s}+\mathbf{H}\mathbf{b})},~~\forall i.
\end{equation}

\subsection{MIMO Point-to-Point VLC}
For illustrative purpose, we study the cases with two RGB LED and detectors, i.e. $N_t=N_r=2$ and for simplicity we assume $K=3N_t=3N_r$. If only one RGB LED is used, the joint transceiver design and offset optimization problem degrades to a transceiver design problem with a fixed bias, thus assuming two RGB LEDs offers the simplest case for our study.

For JTOD Case 1$-$Case 5 in the followings, system parameters are chosen as: the constellation size of each color channel $M=4$; signal takes value from a 4-PAM constellation $\{-3,-1,1,3\}$ such that the variance $r=5$; the transmit optical power $P_T=30$; the color ratio vector $\bar{\mathbf{b}}=[1/3~1/3~1/3]^T$; the number of local runs is $N_I=20$; the stopping threshold values is $\epsilon_P=10^{-3}$. The shot noise scaling factor and the thermal noise variance are $\varsigma^2=1$ and $\sigma^2=0.01$ when the results in the cases are generated, but we point out later that these values can vary without effect on the optimized structures of the transceiver and offset.

\begin{list2}

\item \textit{Case 1 (Perfect):} We first give some intuition on the design with an identity channel matrix
\begin{equation*}
\mathbf{H}=\mathbf{I}_6,
\end{equation*}
i.e. no channel correlation or color cross-talks is considered. The following optimized precoder is obtained (results are rounded up to 2 decimal places)
\begin{equation*}
\mathbf{P}^*=\text{diag}([1.67~ 1.67~ 1.67~ 1.67~ 1.67~ 1.67]),
\end{equation*}
the optimized offset is
\begin{equation*}
\mathbf{b}^*=[5.00~ 5.00~ 5.00~ 5.00~ 5.00~ 5.00]^T,
\end{equation*}
and the optimized post-equalizer is
\begin{align*}
\mathbf{G}^*=\text{diag}([0.59~ 0.59~ 0.59~ 0.59~ 0.59~ 0.59]).
\end{align*}
It is expected that equal offset allocation is the optimal.

\item \textit{Case 2 (Blockage):} Assume that the second RGB LED is blocked with neither correlation nor color cross-talk, i.e.
\begin{equation*}
\bar{\mathbf{H}}=
\begin{bmatrix}
1 & 0\\
0 & 0
\end{bmatrix},
~~\tilde{\mathbf{H}}=\mathbf{I}_3.
\end{equation*}

The optimized precoder is
\begin{equation*}
\mathbf{P}^*=
\begin{bmatrix}
   0 &  0     &    0    &     0    &     0    &     10/3\\
   0  & 0   &      0    &     0    &     0   &      0\\
         0     &    10/3  & 0  & 0   &     0   &      0\\
         0     &    0   & 0 &  0   &      0   &      0\\
         0     &    0    &     0   &      10/3 &  0 & 0\\
         0      &   0     &    0   &      0 & 0 &  0
\end{bmatrix}.
\end{equation*}
The optimized offset is
\begin{equation*}
\mathbf{b}^*=[10.00~ 0~ 10.00~ 0~ 10.00~ 0]^T,
\end{equation*}
and the optimized post-equalizer is
\begin{equation*}
\mathbf{G}^*=
\begin{bmatrix}
   0 &  0     &    0    &     0    &     0    &     0\\
   0  & 0   &      0.29    &     0    &     0   &      0\\
         0     &    0  & 0  & 0   &     0   &      0\\
         0     &    0   & 0 &  0   &      0.29   &      0\\
         0     &    0    &     0   &      0 &  0 & 0\\
         0.29      &   0     &    0   &      0 & 0 &  0
\end{bmatrix}.
\end{equation*}
It is expected that no offset power is allocated to the blocked RGB LED, and its signal is nulled by the precoder.

\item \textit{Case 3 (Correlation):} The ChC and CoC parts of channel are chosen as
\begin{equation*}
\bar{\mathbf{H}}=
\begin{bmatrix}
1 & 0.5\\
0.25 & 1
\end{bmatrix},
~~\tilde{\mathbf{H}}=\mathbf{I}_3,
\end{equation*}
where asymmetric channel correlation is considered and no color cross-talk is assumed. More interference is caused from the second LED to the first detector than from the first LED to the second detector. The optimized precoder is
\begin{equation*}
\mathbf{P}^*=\text{diag}([1.77~ 1.56~ 1.77~ 1.56~ 1.77~ 1.56]),
\end{equation*}
The optimized offset is
\begin{equation*}
\mathbf{b}^*=[5.33~ 4.67~ 5.33~ 4.67~ 5.33~ 4.67]^T,
\end{equation*}
and the optimized post-equalizer is
\begin{align*}
&\mathbf{G}^*=\notag\\
&\begin{bmatrix}
    0.64  & -0.32    &     0   &      0    &     0     &    0\\
   -0.18  &  0.73    &     0   &      0    &     0     &    0\\
         0   &      0  &  0.64 &  -0.32  &       0   &      0\\
         0   &      0  & -0.18 &   0.73  &       0   &      0\\
         0   &      0  &       0 &        0  &  0.64 &  -0.32\\
         0   &      0  &       0 &        0  & -0.18 &   0.73
\end{bmatrix}.
\end{align*}
It is observed that more power is allocated to the first RGB LED, which is expected since it cause less interference to the second RGB LED.

\item \textit{Case 4 (Color Cross-talks):} The ChC and CoC parts of channel are chosen as
\begin{equation*}
\bar{\mathbf{H}}=\mathbf{I}_2,
~~\tilde{\mathbf{H}}=\begin{bmatrix}
0.95   & 0.05 & 0\\
0.05   & 0.90 & 0.05\\
0      & 0.05 & 0.95
\end{bmatrix},
\end{equation*}
where medium CoC is considered and no channel correlation is assumed. The optimized precoder is
\begin{equation*}
\mathbf{P}^*=\text{diag}([1.67~ 1.67~ 1.67~ 1.67~ 1.67~ 1.67~]).
\end{equation*}
The optimized offset is
\begin{equation*}
\mathbf{b}^*=[5.00~ 5.00~ 5.00~ 5.00~ 5.00~ 5.00]^T,
\end{equation*}
and the optimized post-equalizer is
\begin{align*}
&\mathbf{G}^*=\notag\\
&\begin{bmatrix}
    0.63     &    0 &   -0.03   &       0  &   0   &       0\\
         0  &   0.63  &        0  &  -0.03   &       0 &    0\\
   -0.03     &     0  &   0.67   &       0 &   -0.03  &        0\\
         0  &  -0.03  &        0  &   0.67   &       0  &  -0.03\\
    0     &    0 &   -0.03    &     0   &  0.63   &       0\\
         0 &    0 &         0  &  -0.03   &       0  &   0.63
\end{bmatrix}.
\end{align*}
An interesting observation is that the optimal $\mathbf{P}^*$ and $\mathbf{b}^*$ are the same with the no cross-talks case.

\item \textit{Case 5 (Correlation+Color Crosstalks):} The ChC and CoC parts of channel are chosen as
\begin{equation*}
\bar{\mathbf{H}}=
\begin{bmatrix}
1 & 0.50\\
0.25 & 1
\end{bmatrix}
~~\tilde{\mathbf{H}}=
\begin{bmatrix}
0.95   & 0.05 & 0\\
0.05   & 0.90 & 0.05\\
0      & 0.05 & 0.95
\end{bmatrix},
\end{equation*}
where both medium color cross-talk and asymmetric channel correlation are assumed. The optimized precoder is
\begin{equation*}
\mathbf{P}^*=
\begin{bmatrix}
   1.80  &  0 &    0  &  0 &   0  &   0\\
    0  &  1.54  &   0  &  0 &   0  &   0\\
    0  &   0  &   0 &   1.80  &   0 &   0\\
    0  &   0 &   1.54  &   0 &  0  &   0\\
   0  &   0 &   0  &   0  &  1.80  &  0\\
    0  &  0 &   0 &    0  &   0  &  1.54
\end{bmatrix}.
\end{equation*}
The optimized offset is
\begin{equation*}
\mathbf{b}^*=[5.41~ 4.59~ 5.41~ 4.59~ 5.41~ 4.59]^T,
\end{equation*}
and the optimized post-equalizer is
\begin{align*}
&\mathbf{G}^*=\notag\\
&\begin{bmatrix}
    0.66  & -0.33 &  -0.04 &   0.02  &  0  &  0\\
   -0.19 &   0.77 &  -0.01 &  -0.04   & 0 &   0\\
    0.01  &  -0.04 &  -0.20  &  0.82 &  0.01 &  -0,04\\
   -0.04 &   0.02 &   0.70 &  -0.35 &  -0.04 &   0.02\\
    0  &  0 &  -0.04  &  0.02 &   0.66 &  -0.33\\
    0  &  0 &  0.01  & -0.04 &  -0.19  &  0.77
\end{bmatrix}.
\end{align*}
\end{list2}
Note that it is more practical to assume medium-to-low CoC in practice as long as the quality of color filters at the Rx is acceptable. By simulations it is seen that the optimal $\mathbf{P}^*$ and $\mathbf{b}^*$ remain the same if high CoC is considered ($\xi=0.20$). This observation, along with the one in case 4, implied that the optimal structure of the transmitter and offset do not change with CoC.

Then we test the symbol error rate (SER) performance of the case 1 and case 5 across a particular range of noise and power parameters in Fig. 5$-$Fig. 7, on an Intel Core TM i5-3337U (1.80GHz) processor. With $\varsigma^2=1$ and $P_T=30$ fixed, the SER performance comparison among the three schemes for varying thermal noise variance is shown in Fig. 5. With $\sigma^2=0.01$ and $P_T=30$ fixed, the SER performance comparison for varying shot noise scaling factor is shown in Fig. 6. With $\sigma^2=0.01$ and $\varsigma^2=1$ fixed, the SER performance comparison for varying optical power is shown in Fig. 7. It is seen from these figures that for Case 5 the JTOD has large performance gain, but with the perfect channel assumed in Case 1 simpler schemes work out about equally well. The JTOD also outperforms the rest of schemes for Case 3 and Case 4, although results are not shown here. For Case 2, all schemes are not giving a satisfactory result for the range of parameters tested. In Fig. 8, the convergence behavior of our algorithm is plotted, when $\sigma^2=0.01$, $\varsigma^2=1$ and $P_T=30$ are associated. The gradient projection based procedure converges fast, although it may fail on condition that extreme values of parameters are set, e.g. when the channel matrix is very badly conditioned. The averaged processing time over $100$ iterations is summarized in Table 1, which shows good scalability of our method.

\newcounter{FandQ2}
\begin{figure*}[hb]
\hrulefill \setcounter{FandQ}{\value{equation}}
\setcounter{equation}{10}
\begin{table}[H]
\label{tab:fonts}
\centering
%\begin{center}
\begin{tabular}{|c|c|c|c|c|c|c|c|c|c|} %% this creates two columns
\hline
\rule[-1ex]{0pt}{3.5ex}  \backslashbox {} & $N_t=2$ & $N_t=3$ & $N_t=4$ & $N_t=6$ & $N_t=8$ & $N_t=10$ & $N_t=15$ & $N_t=20$ & $N_t=25$   \\
\hline
\rule[-1ex]{0pt}{3.5ex}   \text{Processing Time (sec)}& 1.728 & 1.744 & 1.906 & 2.345 & 2.808  & 3.310 & 3.348 & 5.529 & 10.450\\
\hline
\end{tabular}
%\end{center}
\captionsetup{justification=centering}
\caption{MATLAB running time for each iteration.}\label{table2}
\end{table}
\setcounter{equation}{\value{FandQ}}
\end{figure*}

\begin{figure}[htbp]
\centerline{\includegraphics[width=1.1\columnwidth]{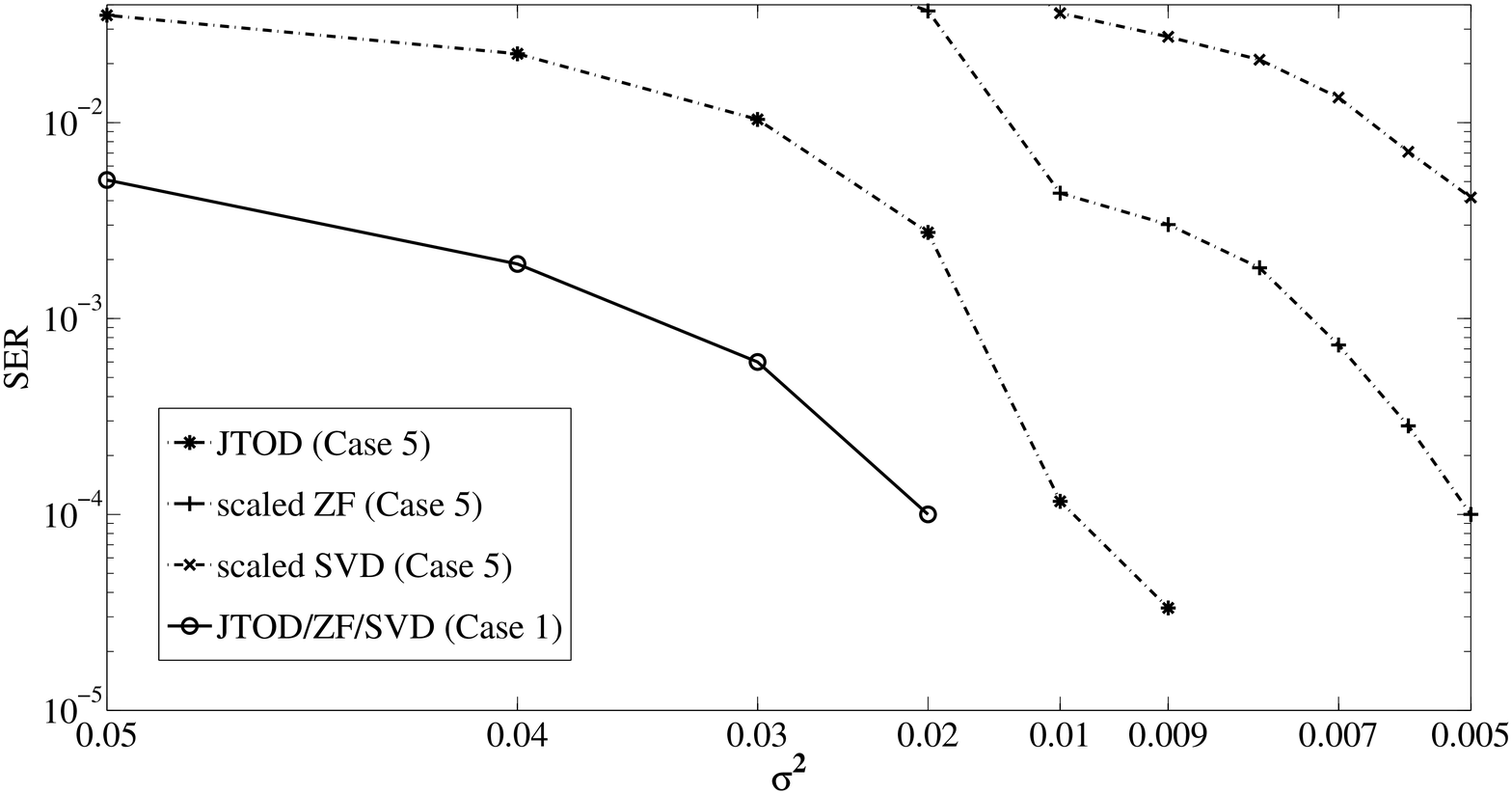}}
\vspace{-0.1in}
\caption{SER performance comparison between schemes for case 1 and case 5 channels with varying $\sigma^2=1$.}
\end{figure}

\begin{figure}[htbp]
\centerline{\includegraphics[width=1.1\columnwidth]{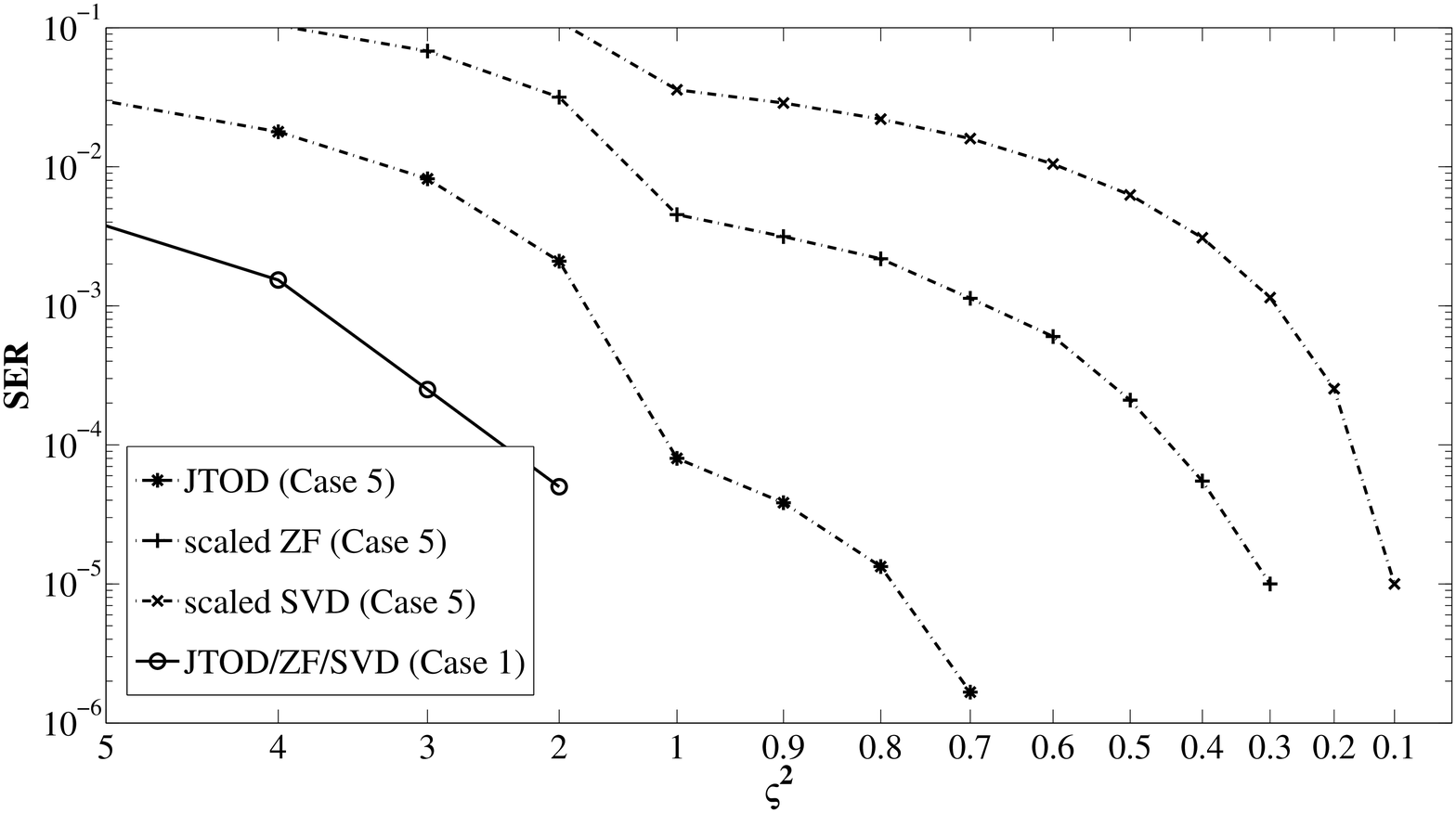}}
\vspace{-0.1in}
\caption{SER performance comparison between schemes with Case 1 and Case 5 channels varying $\varsigma^2$; MIMO P2P.}
\end{figure}

\begin{figure}[htbp]
\centerline{\includegraphics[width=1.1\columnwidth]{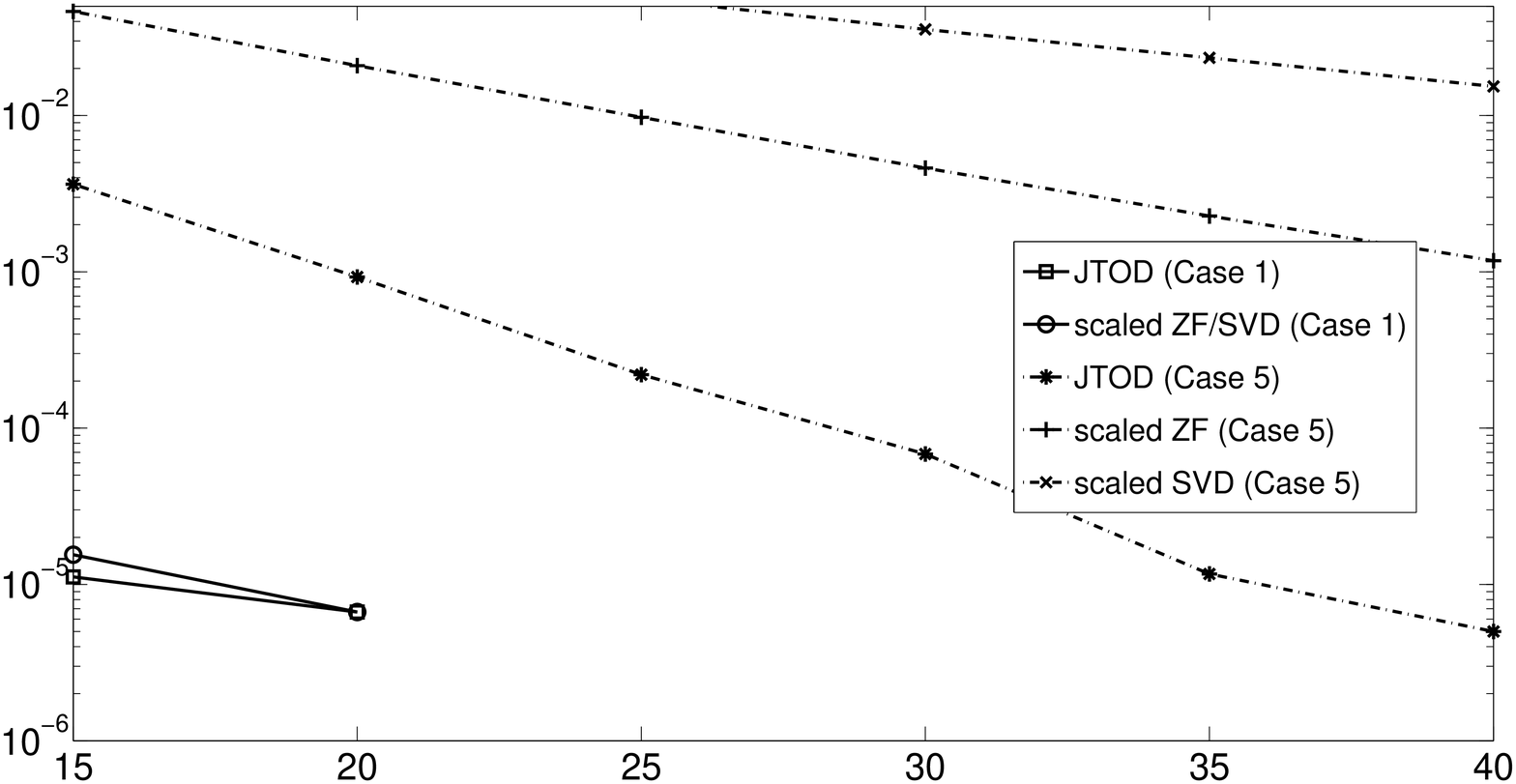}}
\vspace{-0.1in}
\caption{SER performance comparison between schemes with Case 1 and Case 5 channels varying $P_T$; MIMO P2P.}
\end{figure}

\begin{figure}[htbp]
\centerline{\includegraphics[width=1.1\columnwidth]{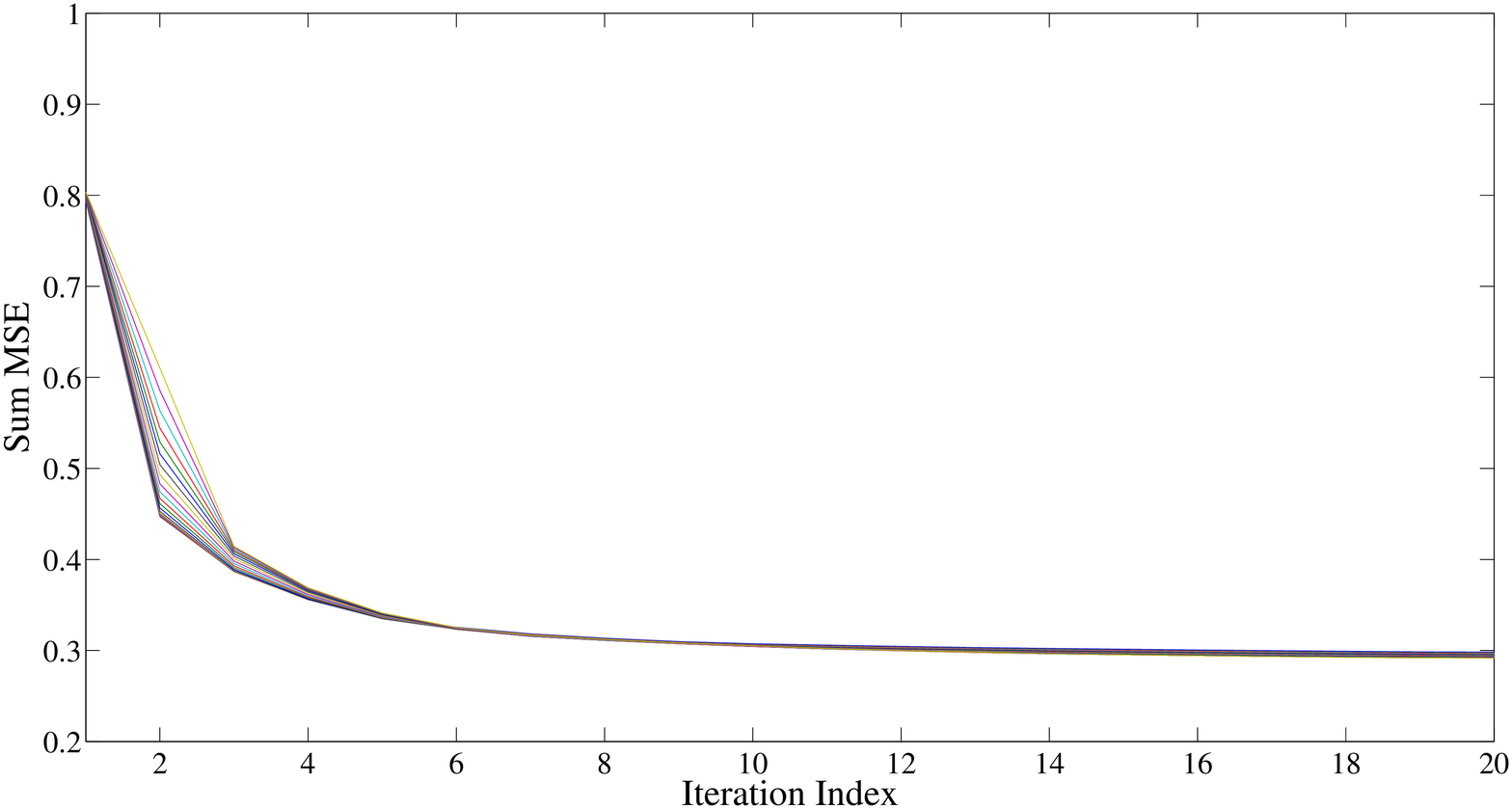}}
\vspace{-0.1in}
\caption{Convergence behavior of the JTOD algorithm, $\sigma^2=0.01$, $\varsigma^2=1$.}
\end{figure}

\subsection{MU-MISO Downlink VLC}

We compare our scheme with the scaled ZF scheme for the downlink problem under the sum MSE constraint as shown by \eqref{39}. A small size example is studied with $N_t=4$ white LEDs serving $K=4$ users respectively. Symbols of each link take values from the same 4-PAM constellation $\{-3,-1,1,3\}$ assumed in previous sections. We only consider the scenario where LEDs are separated enough such that only adjacent ones interfere with each other, and the following normalized channel matrix is used in the simulations
\begin{align*}
&\mathbf{H}=
\begin{bmatrix}
   1 &  0.25    &  0 & 0 \\
   0.25    &  1   &  0.25  & 0 \\
   0    &  0.25  & 1 & 0.25\\
   0  & 0 & 0.25 & 1
\end{bmatrix}.
\end{align*}

In Fig. 8 - Fig. 10, the performance comparison between the JTOD and the scaled ZF across particular noise and power parameter ranges are shown. With $\varsigma^2=0.5$ and $P_T=30$, the SER performance comparison for varying thermal noise variance is shown in Fig. 8. With $\varsigma^2=0.01$ and $P_T=20$, the SER performance comparison for varying shot noise scaling factor is shown in Fig. 9. With $\varsigma^2=0.01$ and $\varsigma^2=0.5$, the SER performance comparison for varying optical power is shown in Fig. 10. It is observed that the JTOD scheme outperforms the scaled ZF scheme non-trivially for most cases, but the performance gain can shrink when the signal is much stronger than noise, e.g. the end of the bottom two lines of Fig. 9 and Fig. 10. In Figure 11, we compare the worst-case SER performance of the robust and the regular JTOD under downlink MU-MISO channel when the noise variance $\sigma^2=0.001$ and $\sigma^2=0$ respectively (when there is no noise, CSI error is still possible because of, e.g. quantization error). The robust scheme is shown to offer better worst-case SER performance over the regular scheme over a range of channel uncertainty values.

%If the noise variance is too large or the average SER is concerned, the robust scheme does not guarantee a better performance over the regular counterpart, and the performance under these cases are not included.

\begin{figure}[H]
\centerline{\includegraphics[width=1.1\columnwidth]{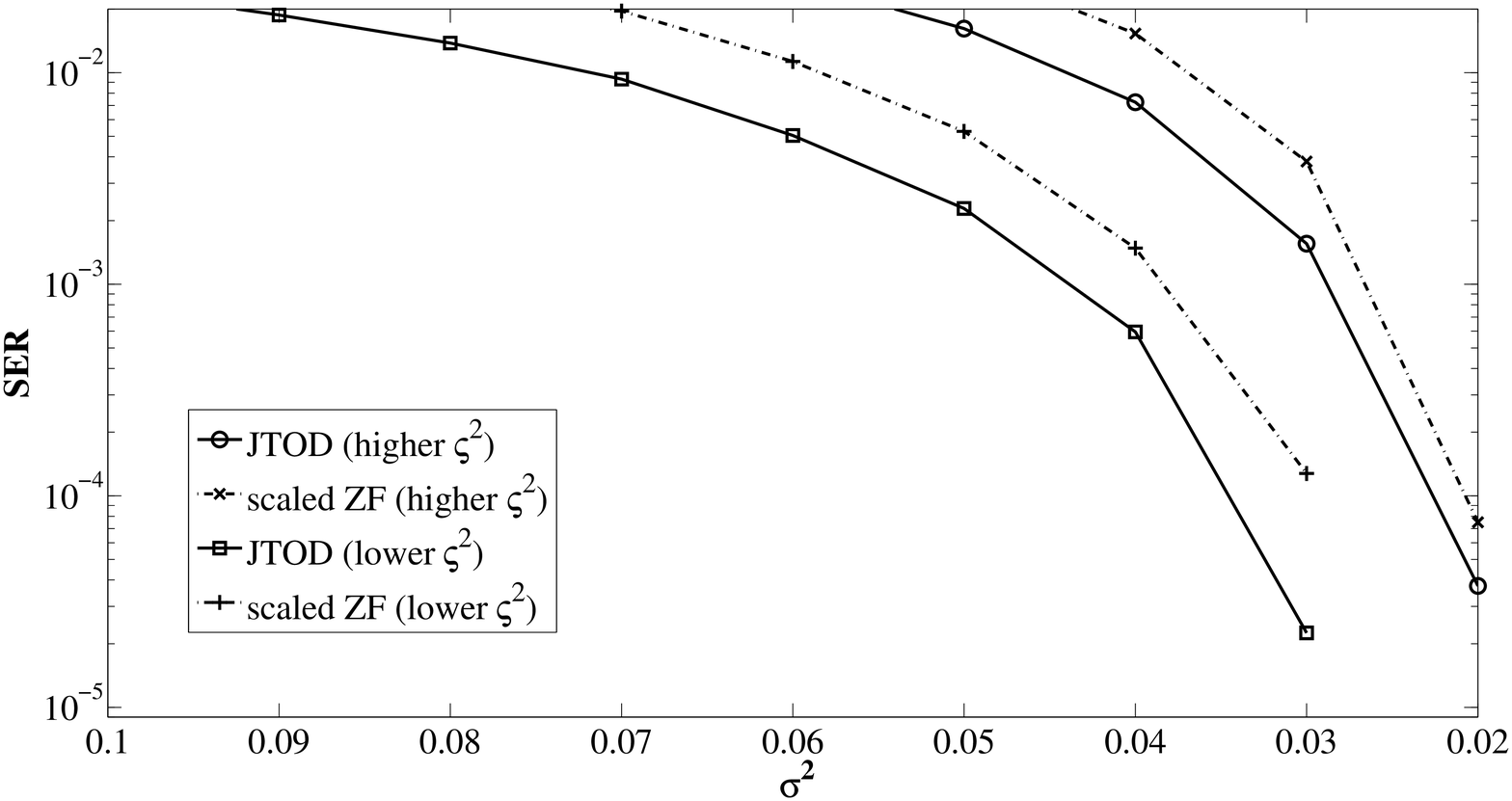}}
\vspace{-0.1in}
\caption{SER performance comparison between JTOD and scaled ZF varying $\sigma^2$; downlink.}
\end{figure}

\begin{figure}[H]
\centerline{\includegraphics[width=1.1\columnwidth]{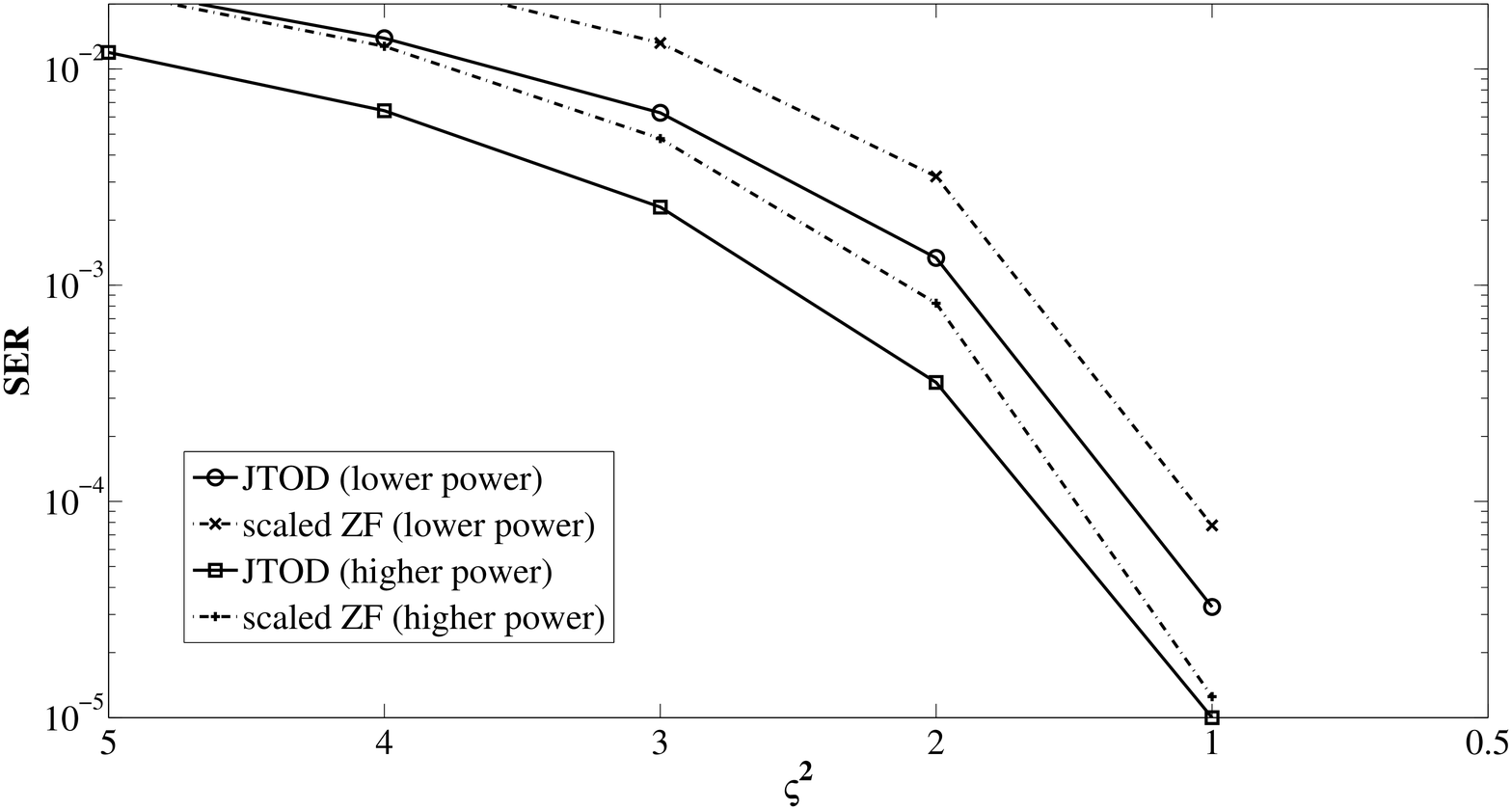}}
\vspace{-0.1in}
\caption{SER performance comparison between JTOD and scaled ZF varying $\varsigma^2$; downlink.}
\end{figure}

\begin{figure}[H]
\centerline{\includegraphics[width=1.1\columnwidth]{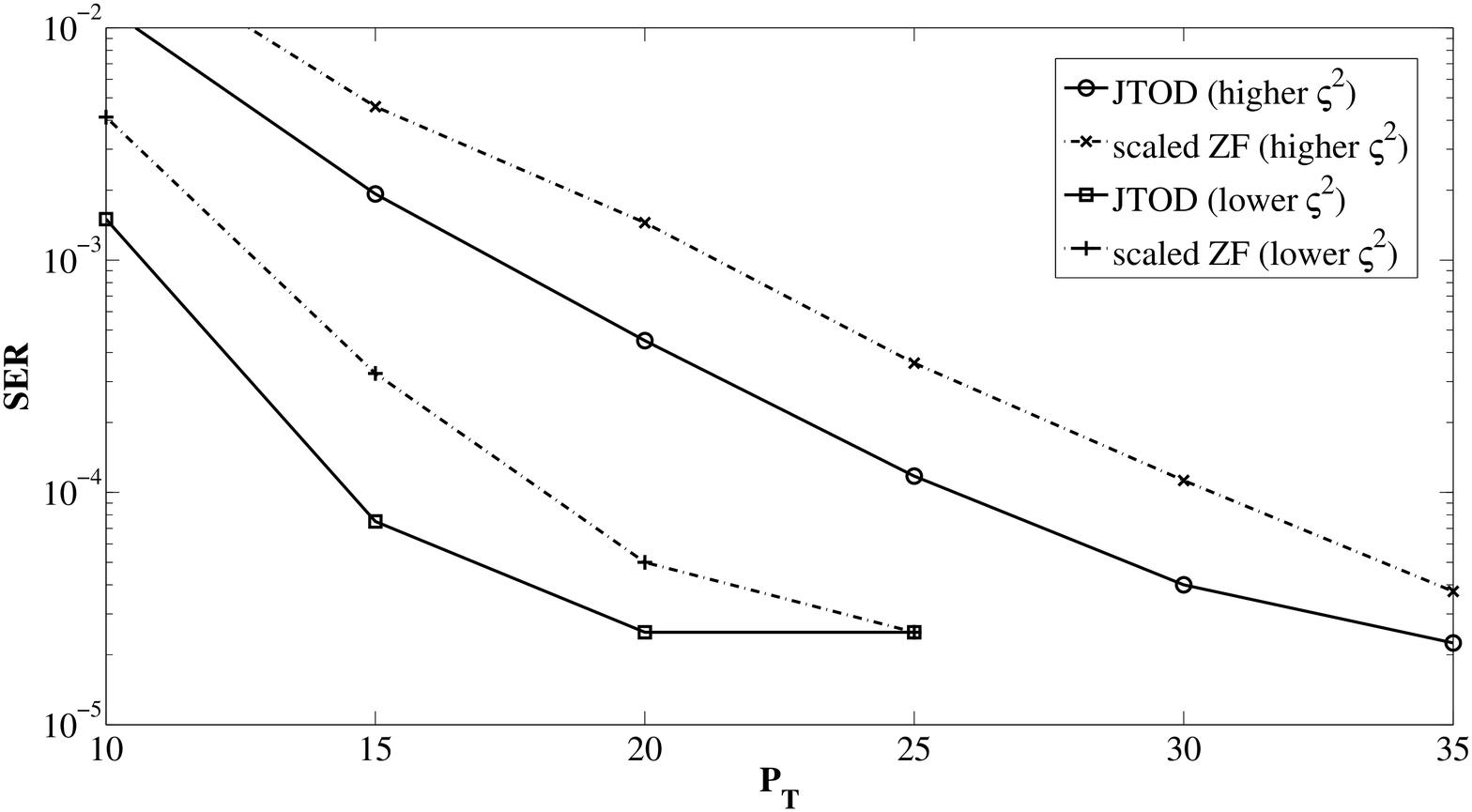}}
\vspace{-0.1in}
\caption{SER performance comparison between JTOD and scaled ZF varying $P_T$; downlink.}
\end{figure}

\begin{figure}[H]
\centerline{\includegraphics[width=1.1\columnwidth]{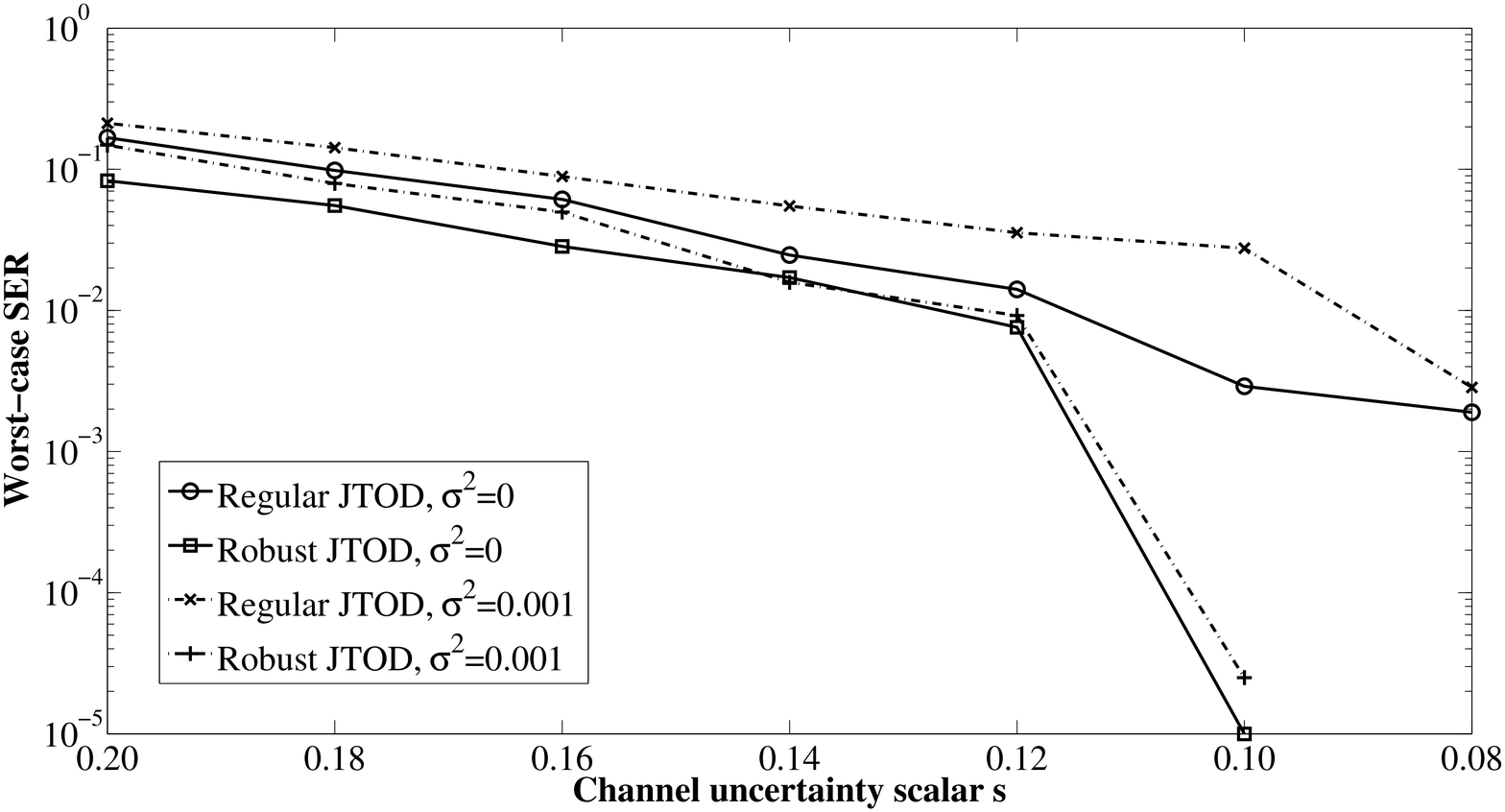}}
\vspace{-0.1in}
\caption{SER performance comparison between the robust and regular JTOD for the downlink MU-MISO channel.}
\end{figure}

\section{Conclusions}\label{sec6}
We have proposed a joint transceiver and offset design method, termed JTOD, for visible light communications in this paper. Both the single user point-to-point and the multi-user uplink and downlink systems are considered, with lighting constraints and input-dependent noise taken into account. Case by case studies are provided for different scenarios, e.g. perfect channel, blockage channel, color-crosstalk channel, and correlated channel and the best structures of the transceivers and offsets obtained correspondingly are presented. These structures when the channel is fixed are observed to be unique for different noise variances and only scale with the optical power. For the point-to-point systems, the JTOD outperforms conventional transceiver design methods with fixed offset, such as ZF and SVD based ones. For the multi-user downlink systems, the performance of ZF approaches the JTOD when the signal is much stronger than noise. We also prove in this paper that different from the RFC counterpart, the downlink-uplink duality does not hold for such type of design. Further, the design of robust JTOD is included, when perfect CSI is not available.

% Change=R3.Minor.10
\appendices
\section{Proof of Proposition 1}
Define
\begin{equation}
g(\mathbf{P})={tr\bigg(\frac{1}{r}\mathbf{I}+\mathbf{P}^T\mathbf{H}^T\big(\varsigma^2\sigma^2\text{diag}(\mathbf{H}\mathbf{b})+\sigma^2\mathbf{I}\big)^{-1}\mathbf{H}\mathbf{P}\bigg)^{-1}} ,
\end{equation}
\begin{equation}
f(\alpha)=g\big(\alpha\mathbf{P}_1+(1-\alpha)\mathbf{P}_2\big),
\end{equation}
where $\mathbf{P}_1$ and $\mathbf{P}_2$ are $3N_t\times 3N_t$ real valued matrices and
\begin{equation}
\boldsymbol{\Sigma}=(\varsigma^2\sigma^2\text{diag}(\mathbf{H}\mathbf{b})+\sigma^2\mathbf{I})^{-1}. \end{equation}
To simplify notations, define
\begin{equation}
\mathbf{A}=\frac{1}{r}\mathbf{I}+(\alpha\mathbf{P}_1^T+(1-\alpha)\mathbf{P}_2^T)\mathbf{H}^T\boldsymbol{\Sigma}
\mathbf{H}(\alpha\mathbf{P}_1+(1-\alpha)\mathbf{P}_2),
\end{equation}
thus $f(\alpha)={\it tr}(\mathbf{A})$. The second order derivative of $f(\alpha)$ with respect to $\alpha$ is calculated by
\begin{equation}\label{Eq_add_1}
\frac{\partial^2f(\alpha)}{\partial^2\alpha}=2{\rm Tr}(\mathbf{A}^{-1}\frac{\partial\mathbf{A}}{\partial\alpha}\mathbf{A}^{-1}\frac{\partial\mathbf{A}}{\partial\alpha}\mathbf{A}^{-1})
-{\rm Tr}(\mathbf{A}^{-1}\frac{\partial^2\mathbf{A}}{\partial^2\alpha}\mathbf{A}^{-1}),
\end{equation}
where we find
\begin{equation} \label{Eq_add}
\begin{split}
%\mathbf{A} &=\frac{1}{r}\mathbf{I}+[\alpha\mathbf{P}_1+(1-\alpha)\mathbf{P}_2]^T\mathbf{H}^T\boldsymbol{\Sigma}\mathbf{H}[\alpha\mathbf{P}_1+(1-\alpha)\mathbf{P}_2],\\
\frac{\partial\mathbf{A}}{\partial\alpha} &= 2 \alpha \mathbf{P}^T_3 \mathbf{H}^T\boldsymbol{\Sigma}\mathbf{H} \mathbf{P}_3 + \mathbf{P}^T_3 \mathbf{H}^T\boldsymbol{\Sigma}\mathbf{H} \mathbf{P}_2 + \mathbf{P}^T_2 \mathbf{H}^T\boldsymbol{\Sigma}\mathbf{H} \mathbf{P}_3, \\
\frac{\partial^2\mathbf{A}}{\partial^2\alpha} &= 2 \mathbf{P}^T_3 \mathbf{H}^T\boldsymbol{\Sigma}\mathbf{H} \mathbf{P}_3,
\end{split}
\end{equation}
where $\mathbf{P}_3= \mathbf{P}_1-\mathbf{P}_2$.
Substituting \eqref{Eq_add} into \eqref{Eq_add_1}, we see that $\frac{\partial^2f(\alpha)}{\partial^2\alpha}$ cannot be always positive. The sign of $\frac{\partial^2f(\alpha)}{\partial^2\alpha}$ actually depends on the value of $\alpha$, $\mathbf{H}$, $\mathbf{P}_1$, and $\mathbf{P}_2$. We thus conclude that $f(\alpha)$ is not convex in $\mathbf{P}$. We provide one example to verify that the second derivative can be negative: choose $\alpha=r=\varsigma^2=\sigma^2=1$, $\mathbf{P}_1=\mathbf{0}$, $\mathbf{P}_2=\mathbf{I}_{3N_t}$, $\mathbf{H}=\mathbf{I}_{3N_t}$ and $\mathbf{b}=\mathbf{1}_{3N_t}$. A similar procedure shows the non-convexity of the objective function w.r.t. $\mathbf{b}$.

\section{Proof of Proposition 2}
From \eqref{37}, the maximum possible sum MSE over the channel uncertainty region of downlink VLC can be transformed into
\begin{align}
\text{MSE}^{\text{DL}}_{max}&=
\max_{\mathbf{h}_k\in\mathcal{U}_k(\boldsymbol{\delta}_k)} r\sum_{k=1}^K ||g_k\mathbf{h}_k^T\mathbf{P}-\mathbf{e}^T||^2+||\sigma\mathbf{g}||^2\notag\\
&+\sum_{k=1}^K\varsigma^2\sigma^2g_k^2\mathbf{h}_k^T\mathbf{b}.
\end{align}
Introducing slack variables $\lambda_0$, $\lambda_k$, $\tilde{\lambda}_k$, and transform the optimization problem \eqref{51} into
\begin{equation}
\begin{aligned}
& \underset{\mathbf{P},\mathbf{b},\boldsymbol{\lambda},\boldsymbol{\tilde{\lambda}},\boldsymbol{\lambda}'}{\text{min}}
& & {\lambda_0+\sum_{k=1}^K}(\lambda_k+\tilde{\lambda}_k)\\
& \text{s.t.}
&& ||\sigma\mathbf{g}||^2\leq \lambda_0,\\
&&&
r||g_k\mathbf{h}_k^T\mathbf{P}-\mathbf{e}_k^T||^2\leq\lambda_k,~\mathbf{h}_k\in\mathcal{U}_k(\boldsymbol{\delta}_k),\\
&&&
\varsigma^2\sigma^2g_k^2\mathbf{h}_k^T\mathbf{b}\leq \tilde{\lambda}_k,\qquad\mathbf{h}_k\in\mathcal{U}_k(\boldsymbol{\delta}_k),\\
&&& \text{abs}(\mathbf{P})\boldsymbol{\delta}-\mathbf{b}\leq \mathbf{0},\\
&&& \mathbf{1}^T\mathbf{b}= \beta P_T.\label{66}
\end{aligned}
\end{equation}
For the second constraint, we apply the Schur complement lemma to transform into the following form
\begin{align}
&\begin{bmatrix}
\lambda_k & g_k\hat{\boldsymbol{h}}_k^T\mathbf{P}-\mathbf{e}_k\\
(g_k\hat{\boldsymbol{h}}_k^T\mathbf{P}-\mathbf{e}_k)^T & r^{-1}\mathbf{I}
\end{bmatrix}\notag\\
\succeq
&\begin{bmatrix}
0 & \boldsymbol{\delta}_k^Tg_k\mathbf{P}\\
g_k\mathbf{P}^T\boldsymbol{\delta}_k & \mathbf{O}
\end{bmatrix}, ~~||\boldsymbol{\delta}_k||\leq \rho_k\label{67}
\end{align}
where $\mathbf{O}$ is a $K\times K$ all-zero matrix. As there are an infinite number of vector $\boldsymbol{\delta}_k$'s that satisfy $||\boldsymbol{\delta}_k||\leq \rho_k$, \eqref{67} still contains an infinite number of equations. To transform it into a single constraint, we apply the following lemma
\begin{mylemma}
Given matrices $\mathbf{B}$, $\mathbf{C}$, $\mathbf{A}$ with $\mathbf{A}=\mathbf{A}^H$, the semi-infinite linear matrix equation of the form of
\begin{equation}
\mathbf{A}\geq \mathbf{B}^H\mathbf{X}\mathbf{C}+\mathbf{C}^H\mathbf{X}^H\mathbf{B},~~\forall \mathbf{X}:||\mathbf{X}||_F\leq \rho,
\end{equation}
holds if and only if there exists $\nu\geq 0$ such that
\begin{equation}
\begin{bmatrix}
\mathbf{A}-\nu\mathbf{C}^H\mathbf{C} & -\rho\mathbf{B}^H\\
-\rho\mathbf{B} & \nu\mathbf{I}
\end{bmatrix}\succeq 0
\end{equation}
\end{mylemma}
By choosing
\begin{equation}
\mathbf{A}=
\begin{bmatrix}
\lambda_k & g_k\hat{\mathbf{h}}_k^T\mathbf{P}-\mathbf{e}_k^T\\
(g_k\hat{\mathbf{h}}_k^T\mathbf{P}-\mathbf{e}_k^T)^T & r^{-1}\mathbf{I}
\end{bmatrix},
\end{equation}
\begin{equation}
\mathbf{B}=[\mathbf{0}~~g_k\mathbf{P}],~~\mathbf{C}=[-1~~ \mathbf{0}],~~\mathbf{X}=\boldsymbol{\delta}_k,
\end{equation}
thus, equation \eqref{67} is transformed into
\begin{equation}
\begin{bmatrix}
    \lambda_k-\nu_k  & g_k\hat{\mathbf{h}}_k^T\mathbf{P}-\mathbf{e}_k^T &  \mathbf{0}^T \\
   (g_k\hat{\mathbf{h}}_k^T\mathbf{P}-\mathbf{e}_k^T)^T &   r^{-1}\mathbf{I} &  -\rho_kg_k\mathbf{P}^T \\
   \mathbf{0}  & -\rho_kg_k\mathbf{P} &  \nu_k\mathbf{I}
\end{bmatrix}\succeq 0,~\forall k\\
\end{equation}
\begin{equation}
\nu_k \geq 0,~\forall k,
\end{equation}
where $\nu_k$'s are auxiliary scalar variables. Similarly, we can prove that the third constraint of \eqref{66}
\begin{equation}
\varsigma^2\sigma^2g_k^2\mathbf{h}_k^T\mathbf{b}\leq \tilde{\lambda}_k,~\mathbf{h}_k\in\mathcal{U}_k(\boldsymbol{\delta}_k)
\end{equation}
is equivalent to
\begin{equation}
\begin{bmatrix}
    \tilde{\lambda}_k-\varsigma^2\sigma^2g_k^2\hat{\mathbf{h}}^T\mathbf{b}-\tilde{\nu}_k(\frac{1}{2}\varsigma^2\sigma^2g_k^2)^2  & -\rho_k\mathbf{b}^T\\
   -\rho_k\mathbf{b} &  \tilde{\nu}_k\mathbf{I}
\end{bmatrix}\succeq 0,~\forall k\\
\end{equation}
\begin{equation}
\tilde{\nu}_k\geq 0,~~\forall k,
\end{equation}
where $\tilde{\nu}_k$'s are auxiliary scalar variables.

\end{document}